\documentclass[article,onecolumn, draftclsnofoot]{IEEEtran}           
\usepackage{hyperref}
\usepackage{bbm}
\usepackage{amsmath}
\usepackage{mathrsfs}
\usepackage{xcolor}
\usepackage{graphicx}
\usepackage{amsfonts}
\usepackage[english]{babel}
\usepackage{amsthm}
\usepackage{epstopdf}

\newtheorem{theorem}{Theorem}

\newtheorem{definition}{Definition}
\newtheorem{lemma}{Lemma}

\newtheorem{lemma*}{Lemma}
\newtheorem{assumption}{Assumption}
\newtheorem{proposition}{Proposition}

\newcommand{\mA}{\mathcal{A}}
\newcommand{\mX}{\mathcal{X}}

\newcommand{\mN}{[N]}
\newcommand{\cP}{\mathcal{P}}

\renewcommand{\hat}{\widehat}

\def\lpr{\left\{}
\def\rpr{\right\}}

\newcommand{\mP}{\mathbb{P}}
\newcommand{\mE}{\mathbb{E}}	

\newcommand{\eq}[1]{\begin{align}#1\end{align}}
\newcommand{\seq}[1]{\begin{subequations}#1\end{subequations}}

\newcommand{\lb}[1]{\left\{ \begin{array}{ll} #1 \end{array} \right.}

\newcommand{\E}{\mathbb{E}}
 \newcommand{\nn}{\nonumber}
\newcommand{\cX}{\mathcal{X}}
\newcommand{\cZ}{\mu^G}

\newcommand{\cA}{\mathcal{A}}

\newcommand{\tsigma}{\tilde{\sigma}}
\newcommand{\hmu}{\hat{\mu}}

\newcommand{\cH}{\mathcal{H}}
\newcommand{\tgamma}{\tilde{\gamma}}

\newcommand{\defeq}{\buildrel\triangle\over =}

\newcommand{\pushright}[1]{\ifmeasuring@ #1 \else\omit\hfill$\displaystyle#1$\fi\ignorespaces}
\newcommand{\pushleft}[1]{\ifmeasuring@ #1 \else\omit$\displaystyle#1$\hfill\fi\ignorespaces}

\usepackage{xspace}
\usepackage[acronym,nogroupskip,nonumberlist,nopostdot]{glossaries}
\glsdisablehyper

\newacronym{mdp}{MDP}{Markov decision process}
\newacronym{ne}{NE}{Nash equillibrium}
\newacronym{mfe}{MFE}{mean-field equillibrium}
\newacronym{mpe}{MPE}{Markov perfect equillibrium}
\newacronym{mfg}{MFG}{mean-field game}
\newacronym{rl}{RL}{reinforcement learning}
\newacronym{marl}{MARL}{multi-agent reinforcement learning}
\newacronym{iot}{IoT}{Internet of Things}
\newacronym{ssg}{SSG}{Stackelberg security game}
\newacronym{pse}{PSE}{perfect Stackelberg equilibrium}
\newacronym{mpse}{MPSE}{Markov \gls{pse}}
\newacronym{pomdp}{POMDP}{partially observed \gls{mdp}}
\newacronym{mc}{MC}{Monte Carlo}
\newacronym{pbe}{PBE}{perfect Bayesian equilibrium}
\newacronym{gmfg}{GMFG}{graphon mean-field game}
\newacronym{gmfe}{GMFE}{graphon mean-field equillibrium}
\newacronym{spe}{SPE}{sub-game perfect equillibrium}
\newacronym{mkv}{MKV}{McKean-Vlasov}

\newcommand{\mdp}{\gls{mdp}\xspace}

\newcommand{\mfe}{\gls{mfe}\xspace}
\newcommand{\mpe}{\gls{mpe}\xspace}

\newcommand{\gmfg}{\gls{gmfg}\xspace}
\newcommand{\gmfe}{\gls{gmfe}\xspace}

\newcommand{\mkv}{\gls{mkv}\xspace}

\begin{document}
\title{Master equation of discrete time graphon mean field games and teams}
\author{Deepanshu Vasal, Rajesh Mishra and Sriram Vishwanath\thanks{Deepanshu Vasal is with Department of Electrical and Computer Engineering at Northwestern University.} 
\thanks{Rajesh Mishra and Sriram Vishwanath are with Department of Electrical and Computer Engineering at University of Texas, Austin.}
\thanks{Part of the paper was presented at~\cite{VaMiVi21}.}}% <-this % stops a space

\maketitle

\begin{abstract}
		In this paper, we present a sequential decomposition algorithm equivalent of Master equation to compute \gmfe of \glspl{gmfg} and graphon optimal Markovian policies (GOMPs) of graphon mean field teams (GMFTs). We consider a large population of players sequentially making strategic decisions where the actions of each player affect their neighbors which is captured in a graph, generated by a known graphon. Each player observes a private state and also a common information as a graphon mean-field population state which represents the empirical networked distribution of other players' types. We consider non-stationary population state dynamics and present a novel backward recursive algorithm to compute both \gmfe and GOMP that depend on both, a player's private type, and the current (dynamic) population state determined through the graphon. Each step in computing GMFE consists of solving a fixed-point equation, while computing GOMP involves solving for an optimization problem. We provide conditions on model parameters for which there exists such a \gmfe. Using this algorithm, we obtain the \gmfe and GOMP for a specific security setup in cyber physical systems for different graphons that capture the interactions between the nodes in the system. 
\end{abstract}

\begin{IEEEkeywords}
	Graphon mean-field teams and games, Sequential decomposition,  Signaling, Optimal Markov strategies
\end{IEEEkeywords}

\section{Introduction}
	Interaction of interconnected agents has been an important topic of study for many decades and its relevance has been increasing rapidly with the progress of internet penetration and smartphone devices in our society. The recent decade has seen tremendous technological advancement in the field of networking applications that has led to an unprecedented scale of interaction among people and devices such as in ride sharing platforms, social media apps, cyber-physical systems, autonomous vehicles and drones,  large scale renewable energy, electric vehicles, cryptocurrencies and smart grid systems. For instance, the influence of social networks in the decision making of majority of individuals is a known phenomenon. Most decisions by individuals from which products to buy to whom to vote for are influenced by friends and acquaintances. The emerging empirical evidence on these issues motivates the theoretical study of network effects with strategic and non strategic agents. The analysis, design and control of such systems that involve such interactions embedded in a networked environment could lead to more intelligent and efficient applications, and can enhance our understanding of the mechanics of such interactions. 

Many of the above mentioned applications of interest have following key features: (a) large number of strategic or non strategic players (b) dynamically evolving incomplete information, and (c) an underlying network. When the decision makers are non strategic, one can pose such problems as decentralized stochastic control problems on a network and in general such problems are extremely hard (see~\cite{Wi68,NaMaTe13} and references therein). When it comes to problems with strategic interactions, game theory is a natural choice to model such interactions where the payoffs obtained by individuals depend on the action of her neighbors. A shortcoming of the standard approach to solve dynamic network games with incomplete information is the interdependence of strategies of the players across time. Moreover, as the number of players become large as is the case in many practical scenarios considered here, computing Nash equilibrium becomes intractable. 

\subsection{Relevant Literature}
For the decentralized team problems, Witsenhausen provided a `simple' two stage LQG system~\cite{Wi68} where he showed that linear policies are not optimal and to this day we don't know the optimal policies for that system showing how such simple looking decentralized control system could be extremely hard. Decentralized control systems have been studied extensively in the literature where not too long ago Nayyar et al in~\cite{NaMaTe13} (see references there in) presented a \emph{common agent approach} where showed that a class of decentralized control problems with common information can be posed as a single agent partially observed Markov decision problems and thus in principle can be solved using dynamic programming. Arabneydi and Mahajan posed such a problem with large number of players as Mean field team problems in~\cite{AbMa14} and provided a dynamic programming approach to find optimal Markovian policies for such problems.

There is a huge literature on studying dynamic decision problems when the users are strategic. Maskin and Tirole in~\cite{MaTi01} introduced the concept of \mpe for dynamic games governed by an underlying \mdp. The strategies thus computed depend on the present state and not on the past trajectory of the game. In general, there exists a backward recursive methodology to compute \mpe of the game. Some prominent examples of the application of \mpe include~\cite{ErPa95, BeVa96, AcRo01}. Ericson and Pakes in~\cite{ErPa95} model industry dynamics for firms' entry, exit and investment participation, through a dynamic game with symmetric information, compute its \mpe, and prove ergodicity of the equilibrium process. Bergemann and V{\" a}lim{\"a}ki in~\cite{BeVa96} study a learning process in a dynamic oligopoly with strategic sellers and a single buyer, allowing for price competition among sellers. They study \mpe of the game and its convergence behavior. Acemo\u{g}lu and Robinson in~\cite{AcRo01} develop a theory of political transitions in a country by modeling it as a repeated game between the elites and the poor, and study its \mpe. When players have private types then an appropriate solution concept is perfect Bayesian equilibrium (PBE) and sequential equilibirum (SE). Recently authors in~\cite{VaSiAn16arxiv,VaAn16, Ta17,HeAn20, OuTaTe17} presented backward recursive sequential decomposition methodologies to compute PBE for different classes of dynamic games of incomplete information.

 In large population games, computing \mpe, PBE and SE with the methods specified above becomes intractable. Mean field games (MFG) were introduced in Huang, Malham\'e,and Caines~\cite{HuMaCa06}, and Lasry and Lions~\cite{LaLi07} to model the strategic interactions with large number of players. In such games, the individual agents have minimal impact of the overall outcome of the game and so the agents track a mean distribution of states of other agents rather than their actual states. MFGs is an excellent and a tractable model to study large population dynamic games of incomplete information, and has been shown to be a good approximation of Nash equilibrium (or MPE) of the original game as the number of players grow large (for instance see ~\cite{Caetal15,La16,Fi17,La18,DeLaRa19} and references therein).

Parise and Ozdaglar introduced the notion of graphon games~\cite{PaOz19} to model large population \emph{static} network games, where graphon is generative model of a large random graph inroduced by Lov\"asz in~\cite{Lo12}. Caines and Huang in~\cite{CaHu18} combined the ideas of \mfe and graphon games to define Graphon Mean field games (\glspl{gmfg}) where there are a large number of strategic agents with dynamic incomplete information who interact on an underlying fixed network generated by a known graphon. 
\glspl{gmfg} combine the idea of network games defined through graphons and the mean field framework of describing multi agent homogeneous games and predicting equilibrium in a tractable manner. Large network of nodes interacting with one another can be represented as graphons and mean field games deal with the study of such large interaction among devices and people as agents to analyze such systems to design and understand the behavior of such large scale interactions and their impact on our society. The progress in research in the mean field domain have been restricted to cases where the agents interacted in a perfect homogeneous environment and the interactions between the agents were assumed to be uniform irrespective of the location of the agent in the network. However, in many real world scenarios the population interaction is not uniform and there is a measure of how the agents interacted with each other or in other words, the payoff and the transition to the next state is conditional on the relative position of the agent in the network. Then the mean field distribution would be affected by it and so will the optimum policies and the Nash equilibrium thus generated. The theoretical basis for such a case has been provided in~\cite{CaHu18} which generalizes the idea of mean field games across the population with different levels of interactions through \gmfg. 
 
In this paper, we consider both discounted finite horizon and infinite-horizon dynamic graphon mean-field teams and games where there is a large population of homogeneous players each having a private type. Each player sequentially makes decisions and is affected by other players in its neighborhood through a graphon mean-field population state. Each player has a private type that evolves through a controlled Markov process as a function of the graphon, which only she observes and all players observe a common population state which is the distribution of other players' types. In such games, the graphon mean-field state evolves through McKean-Vlasov forward equation given a policy of the players and the graphon function. The equilibrium policy satisfies the Bellman backward equation, given the graphon mean-field states. Thus to compute equilibrium, one needs to solve the coupled backward and forward fixed-point equation in the graphon mean-field and the equilibrium policy. We propose a sequential decomposition algorithm to compute \glspl{gmfe} and GOMPs by decomposing the problem across time. This algorithm is equivalent to the Master equation of continuous time mean field game~\cite{CaDeLaLi15} that allows one to compute all mean field equilibria (MFE) of the game sequentially.\footnote{Since the publishing an initial version of this paper in~\cite{VaMiVi20}, authors in~\cite{TcCaHu20} have computed a Master's equation for Linear Quadratic Gaussian (LQG) GMFG.}

In order to demonstrate the utility of our algorithm to compute the \gmfe  and GOMP of a graphon mean field game and a team for varying graphons, we consider a cyber-security example of malware spread problem. A cluster of nodes in a network of physical servers get infected by an independent random process. For each node, there is a higher risk of getting infected due to negative externality imposed by other infected players. A graphon function is defined that quantifies the effect of the effect of the state of other nodes in the network on the concerned node. At each time t, a node privately observes its own state and publicly observes the population of infected nodes, based on which it has to make a decision to repair or not. Upon taking an action, the transition of to the  next state is governed by both its individual action and the actions affected by the neighboring agents given by a graphon function. Using our algorithm, we find equilibrium strategies of the players which are observed to be non-decreasing in the healthy population state. Similarly we find optimal Markovian policies for the team problem.

The paper is structured as follows. In Section~\ref{sec:Model}, we present a model of the graphon mean field game and team, followed by some preliminary result from our past research regarding \mpe in strategic dynamic games. In section ~\ref{sec:methodology}, we present our main results where we present algorithm to compute \mpe for both finite and infinite horizon game, and also present existence results. In Section~IV we talk about the existence of GMFE. In Section~V, we
consider graphon team problem and provide a dynamic program to find optimal Markovian policies.
In Section~\ref{sec:example}, we show the simulation results for the cyber-security example assuming different graphons and conclude in Section~\ref{sec:Concl}.

\subsection{Notation}
	We use uppercase letters for random variables and lowercase for their realizations. For any variable, subscripts represent time indices and superscripts represent player identities. We use notation $ -\alpha$ to represent all players other than player $\alpha$ i.e. $ -\alpha = \{1,2, \ldots i-1, i+1, \ldots, N \}$. We use notation $a_{t:t'}$ to represent the vector $(a_t, a_{t+1}, \ldots a_{t'})$ when $t'\geq t$ or an empty vector if $t'< t$. We use $a_t^{-\alpha}$ to mean $(a^1_t, a^2_{t}, \ldots, a_t^{i-1}, a_t^{i+1} \ldots, a^N_{t})$. We use the notation $\sum_x$ to represent both $\sum_x$ and $\int_x$, and the correct usage is determined depending on the space of $x$. We remove superscripts or subscripts if we want to represent the 
	 vector, for example $a_t$  represents $(a_t^1, \ldots, a_t^N) $. We denote the indicator function of any set $A$ by $\mathbbm{1}\{A\}$. For any finite set $\mathcal{S}$, $\mathcal{P}(\mathcal{S})$ represents space of probability measures on $\mathcal{S}$ and $|\mathcal{S}|$ represents its cardinality. We denote by $P^{\sigma}$ (or $E^{\sigma}$) the probability measure generated by (or expectation with respect to) strategy profile $\sigma$. We denote the set of real numbers by $\mathbb{R}$. For a probabilistic strategy profile of players $(\sigma_t^{\alpha})_{i\in [N]}$ where probability of action $a_t^{\alpha}$ conditioned on $\mu^G_{1:t},x_{1:t}^{\alpha}$ is given by $\sigma_t^{\alpha}(a_t^{\alpha}|\mu^G_{1:t},x_{1:t}^{\alpha})$, we use the short hand notation $\sigma_t^{-\alpha}(a_t^{-\alpha}|\mu^G_{1:t},x_{1:t}^{-\alpha})$ to represent $\prod_{j\neq i} \sigma_t^j(a_t^j|\mu^G_{1:t},x_{1:t}^j)$. All equalities and inequalities involving random variables are to be interpreted in \emph{a.s.} sense.

\section{Model and Background}
		\label{sec:Model}
	\subsection{Graphon Mean Field Games and Teams}
		Let us consider a discrete-time large population sequential game with $N$ homogeneous players with $N\to\infty$. The interactions between these $N$ players are captured in a asymptotically infinite network graph represented as a \emph{graphon}. Graphons are bounded symmetric Lebesgue measurable functions $W:\left[0,1\right]^2\to\left[0,1\right]$ which can be represented as weighted graphs on the vertex set $\left[0,1\right]$ such that $\mathbf{G}=\left\{g\left(\alpha,\beta\right):0\leq\alpha,\beta\leq 1\right\}$~\cite{CaHu18}. It is similar to an adjacency matrix defined over a $2$-dimensional plane where each entry in the matrix is the measure of coupling between the agents concerned. 
				
		In each period $t\in[T]$, where $[T]$ represents the time horizon, a player $\alpha\in\left[0,1\right]$ observes a private type $x_t^{\alpha}\in\cX$ and a common observation $\mu_t^G$, then takes an action $a_t^{\alpha}\in\cA $ and receives a reward $R(x_t^{\alpha},a_t^{\alpha},\mu_t^G)$. The common observation is an ensemble of the mean field distributions with respect to all agents $\alpha\in[0,1]$ given as $\mu^G_t = \{\mu_t^{\alpha}\}_{{\alpha}}$ where 	
		\begin{align}
			\mu_t^\alpha\left(x\right)=\mathbbm{P}\left\{x_t^\alpha=x\right\}
		\end{align}	
		with $\sum_{i=1}^{N_x} \mu^{\alpha}_t(i) = 1$. Player $\alpha$'s type evolves as a controlled Markov process,	
		\begin{align}
			\label{eqn:markov}
			x_{t+1}^{\alpha} = \tilde{f}[x_t^{\alpha} , a_t^{\alpha},  \mu_t^G; g^{\alpha} ] + w_t^{\alpha}.
		\end{align}

		The random variables $(w_t^{\alpha})_{{\alpha},t}$ are assumed to be mutually independent across players and across time. We also write the above update of $x_t^{\alpha}$ through a kernel, $x_{t+1}^{\alpha}\sim Q^{\alpha}(\cdot|x_t^{\alpha}, a_t^{\alpha}, \mu^G_t;g^{\alpha})$ which depends on the graphon function $g^\alpha=\left\{g(\alpha,\beta):0\leq\beta\leq1\right\}$. 
		
		The dynamics of the \mdp are governed both by the local information as well as the global dynamics involving the effect of the policy action of other players in the system. The idea of graphon is to capture the effect of the actions of all the other players $\beta\in[0,1]$ on player $\alpha$. In prior mean field research, it was assumed that there is a perfect interaction between the players and also that these interactions were uniform. In~\cite{CaHu18}, they provide a set of differential equations that govern such interactions in the mean field setting. The functions below show how the graphon is used in determining the effect of players on one another. The function $\tilde{f}$ in~\eqref{eqn:markov} is given as		
		\begin{align}
			\label{eq:func1}
			\tilde{f}[x_t^{\alpha} , a_t^{\alpha},  \mu_t^G; g^{\alpha} ] = f_0\left(x_t^{\alpha} , a_t^{\alpha}\right) + f\left[x_t^{\alpha} , a_t^{\alpha},  \mu_t^G; g^{\alpha} \right]
		\end{align}
		where
		\begin{align}
			\label{eq:func2}
			f\left[x_t^{\alpha} , a_t^{\alpha},  \mu_t^G; g^{\alpha} \right] =	\int_{\beta\in[0,1]}\sum_{x^\beta\in\cX}&f\left(x^\alpha_t,a^\alpha_t,x^\beta\right)g\left(\alpha,\beta\right)\mu_t^\beta(x^\beta)d(x^\beta)d\beta 		
		\end{align}
		and $f_0$ represent the local effect of the agent when it takes any action and is independent of the actions taken by other agents. In the case, when the agents do not interact at all i.e. $g(\alpha,\beta)=0$, the markov process reduces only to the function $f_0$ ignoring the degenerate case when $\alpha=\beta$. 
		
		At instant $t$, the player $\alpha$ observes the trajectory $(\mu^G_{1:t},x_{1:t}^{\alpha})$ and takes an action $a_t^{\alpha}$ according to a behavioral strategy $\sigma^{\alpha} = (\sigma_t^{\alpha})^t$, where $\sigma_t^{\alpha}:(\mu^G)^{t}\times\mathcal{X}^t \to \mathcal{P}(\mathcal{A})$. We denote the space of such strategies as $\mathcal{K}^{\sigma}$. This implies $A_t^{\alpha}\sim \sigma_t^{\alpha}(\cdot|\mu^G_{1:t},x_{1:t}^{\alpha})$. We denote $\mathcal{Z}^t$ to be the space of population states $\mu^G_{1:t}$ till time $t$. We denote $\mathcal{H}_t^{\alpha}=\mathcal{Z}^t \times \mathcal{X}^t$ to be set of observed histories $(\mu^G_{1:t},x_{1:t}^{\alpha})$ of player $\alpha$.

		For finite time-horizon game, $\mathbb{G}_{T}$, each player wants to maximize its total expected discounted reward over a time horizon $T$, discounted by discount factor $0<\delta\leq1$, 
		\begin{align}
			J_{Game}^{\alpha,T} :=\E^{\sigma} \left[\sum_{t=1}^T \delta^{t-1} R(X_t^{\alpha},A_t^{\alpha},\mu^G_t;g^{\alpha}) \right].
		\end{align} 

		For the infinite time-horizon game, $\mathbb{G}_{\infty}$, each player wants to maximize its total expected discounted reward over an infinite-time horizon discounted by a discount factor $0<\delta<1$, 
		\begin{align}
			J_{Game}^{\alpha,\infty} :=\E^{\sigma} \left[\sum_{t=1}^\infty \delta^{t-1} R(X_t^{\alpha},A_t^{\alpha},\mu^G_t;g^{\alpha}) \right].
		\end{align} 
		
		Similarly for finite time-horizon team, $\mathbb{T}_{T}$, all players wants to maximize their average total expected discounted reward over a time horizon $T$, discounted by discount factor $0<\delta\leq1$, 
		\begin{align}
			J_{Team}^{T} :=\E^{\sigma} \left[\sum_{t=1}^T \delta^{t-1} \sum_{x_t^{\alpha}}\mu_t^{\alpha}(x_t^{\alpha})R(x_t^{\alpha},A_t^{\alpha},\mu^G_t;g^{\alpha}) \right].
		\end{align} 

		For the infinite time-horizon team, $\mathbb{T}_{\infty}$, each player wants to maximize its total expected discounted reward over an infinite-time horizon discounted by a discount factor $0<\delta<1$, 
		\begin{align}
			J_{Team}^{\infty} :=\E^{\sigma} \left[\sum_{t=1}^\infty \delta^{t-1} \sum_{x_t^{\alpha}}\mu_t^{\alpha}(x_t^{\alpha}) R(x_t^{\alpha},A_t^{\alpha},\mu^G_t;g^{\alpha}) \right].
		\end{align}

	\subsection{Solution concept: \Gls{gmfe}}
% 		The Nash equilibrium (NE) of $\mathbb{G}_{T}$ is defined as strategies $\tsigma = (\tsigma_t^{\alpha})_{\alpha\in\left[0,1\right],t\in\left[T\right]}$ that satisfy, for all $\alpha\in\left[0,1\right]$, 
% 		\begin{align}
% 			\E^{(\tsigma^{\alpha},\tsigma^{-\alpha})}\left[\sum_{t=1}^T \delta^{t-1} R(X_t^{\alpha},A_t^{\alpha},\mu^G_t;g^{\alpha}) \right]\geq \E^{(\sigma^{\alpha},\tsigma^{-\alpha})}\left[\sum_{t=1}^T \delta^{t-1} R(X_t^{\alpha},A_t^{\alpha},\mu^G_t;g^{\alpha})\right],
% 		\end{align} 
		For graphon mean field games, notion of equilibrium is \gmfe~\cite{MaTi01}, which we use in this paper. A \gmfe $(\tsigma)$ satisfies sequential rationality such that for $\mathbb{G}_T$, $\forall\alpha\in\left[0,1\right], t \in \left[T\right], \mu^G_{1:t},x_{1:t}^{\alpha}, {\sigma^{\alpha}}$,
		\begin{align}
			\E^{(\tsigma^{\alpha} \tsigma^{-\alpha})}\left[\sum_{n=t}^T\delta^{n-t} R(X_n^{\alpha},A_n^{\alpha},\mu^G_n;g^{\alpha})|\mu^G_{1:t},x_{1:t}^{\alpha}\right] \geq
			\E^{({\sigma}^{\alpha} \tsigma^{-\alpha})}\left[\sum_{n=t}^T \delta^{n-t} R(X_n^{\alpha},A_n^{\alpha},\mu^G_n;g^{\alpha})|\mu^G_{1:t},x_{1:t}^{\alpha}\right], \;\; \;\;   \label{eq:seqeq2}
		\end{align}
		 \gmfe for $\mathbb{G}_{\infty}$ are defined in a similar way where summation in the above equations is taken such that $T$ is replaced by $\infty$.

		\subsection{Solution concept: Graphon mean field team optimal}

		For graphon mean field teams, we use the notion of optimality as follows. A policy $(\sigma^*)$ is team optimal if for $\forall\alpha\in\left[0,1\right], t \in \left[T\right], \mu^G_{1:t}, {\sigma^{\alpha}}$,
		\begin{align}
			\E^{\sigma^{*}}\left[\sum_{n=t}^T\delta^{n-t}\sum_{x_t^{\alpha}}\mu_t^{\alpha}(x_n^{\alpha}) R(x_n^{\alpha},A_n^{\alpha},\mu^G_n;g^{\alpha})|\mu^G_{1:t}\right] \geq
			\E^{\sigma}\left[\sum_{n=t}^T \delta^{n-t} \sum_{x_t^{\alpha}}\mu_t^{\alpha}(x_n^{\alpha})R(x_n^{\alpha},A_n^{\alpha},\mu^G_n;g^{\alpha})|\mu^G_{1:t}\right], \;\; \;\;   \label{eq:seqeq2}
		\end{align}
		The notion of optimality for $\mathbb{T}_{T}$ are defined in a similar way where summation in the above equations is taken such that $T$ is replaced by $\infty$.
\section{A methodology to compute \glspl{gmfg} }
		\label{sec:methodology}
	In this section, we will provide a backward recursive methodology to compute \glspl{gmfg} for both $\mathbb{G}_{T}$ and $\mathbb{G}_{\infty}$. We will consider Markovian equilibrium strategies of player $\alpha$ which depend on the common information at time $t$, $\mu^G_{t}$, and on its current type $x_t^{\alpha}$.\footnote{Note however, that the unilateral deviations of the player are considered in the space of all strategies.}  Equivalently, player $\alpha$ takes action of the form $A_t^{\alpha}\sim \sigma_t^{\alpha}(\cdot|\mu^G_t,x_t^{\alpha})$. Similar to the common agent approach in~\cite{NaMaTe13}, an alternate and equivalent way of defining the strategies of the players is as follows. We first generate partial function $\gamma_t^{\alpha}:\cX\to\cP(\cA)$ as a function of $\mu^G_t$ through an equilibrium generating function $\theta_t^{\alpha}:\cZ\to(\cX\to\cP(\cA))$ such that $\gamma_t^{\alpha} = \theta_t^{\alpha}[\mu^G_t]$. Then action $A_t^{\alpha}$ is generated by applying this prescription function $\gamma_t^{\alpha}$ on player $\alpha$'s current private information $x_t^{\alpha}$, i.e. $A_t^{\alpha}\sim \gamma_t^{\alpha}(\cdot|x_t^{\alpha})$. Thus $A_t^{\alpha}\sim \sigma_t^{\alpha}(\cdot|\mu^G_{t},x_t^{\alpha}) = \theta_t^{\alpha}[\mu^G_t](\cdot|x_t^{\alpha})$.
	
	For a given prescription function $\gamma^{\alpha}_t = \theta^{\alpha}[\mu^G_t]$, the graphon mean-field $\mu^G_t$ evolves according to the discrete-time McKean Vlasov equation, $\forall y\in\cX$ and $\forall \alpha\in[0,1]$:
	\begin{align}
		\mu^{\alpha}_{t+1}(y) =\sum_{x\in\cX}\sum_{a\in \cA} \mu^{\alpha}_t(x)\gamma^{\alpha}_t(a|x)Q\left(y|x,a,\mu^G_t;g^{\alpha}\right), \label{eq:muG_update}
	\end{align}
	which implies
	\begin{align}
		\mu^{\alpha}_{t+1}&= \phi(\mu^{\alpha}_t,\gamma^{\alpha}_t,\mu^G_t;g^{\alpha})\\
		\mu_{t+1}^G &= \phi(\mu_t^G,\gamma_t;g^{\alpha}) 
	\end{align}

	\subsection{Backward recursive algorithm for $\mathbb{G}_{T}$} \label{sec:fhbr1}
		In this subsection, we will provide a methodology to generate \gmfe of $\mathbb{G}_{T}$ of the form described above. We define an equilibrium generating function $(\theta_t)_{t\in[T]}$, where $\theta_t:\cZ\to\{\cX\to\mathcal{P}(\cA) \}$, where for each $\mu^G_t $, we generate $\tgamma_t = \theta_t[\mu^G_t]$. In addition, we generate a reward-to-go function $(V^{\alpha}_t)_{t\in[T]}$, where $V^{\alpha}_t:\cZ\times\cX\to\mathbb{R}$. These quantities are generated through a fixed-point equation as follows.
		\begin{enumerate}
			\item Initialize $\forall \mu^G_{T+1}, \alpha, x_{T+1}^{\alpha}\in \cX$,
			\begin{align}
				V^{\alpha}_{T+1}(\mu^G_{T+1},x_{T+1}^{\alpha}) \defeq 0.   \label{eq:VT+1}
			\end{align}

			\item For $t = T,T-1, \ldots 1, \ \forall \mu^G_t$, let $\theta_t[\mu^G_t] $ be generated as follows. Set $\tilde{\gamma}_t = (\tgamma_t^{\alpha})_{\alpha \in [0,1]} = \theta_t[\mu^G_t]$, where $\tilde{\gamma}_t$ is the solution of the following fixed-point equation\footnote{We discuss the existence of solution of this fixed-point equation in Section~\ref{sec:exists}}, $\forall \alpha\in[0,1], x_t^{\alpha}\in \cX$,
  			\begin{align}
				 \tilde{\gamma}^{\alpha}_t(\cdot|x_t^{\alpha}) \in  \arg\max_{\gamma_t^{\alpha}(\cdot|x_t^{\alpha})} \E^{\gamma_t^{\alpha}(\cdot|x_t^{\alpha})}\left[ R(X_t^{\alpha},A_t^{\alpha},\mu^G_t;g^{\alpha})+\delta V^{\alpha}_{t+1}(\phi(\mu^G_t,\tgamma_t^{\alpha};g^{\alpha}), X_{t+1}^{\alpha}) | \mu^G_t,x_t^{\alpha}\right] , \label{eq:m_FP}
			\end{align}
 			where expectation in \eqref{eq:m_FP} is with respect to random variable $(A_t^{\alpha},X_{t+1}^{\alpha})$ through the probability measure $\gamma_t(a_t^{\alpha}|x_t^{\alpha})Q^{\alpha}(x_{t+1}^{\alpha}|x_t^{\alpha},a_t^{\alpha},\mu^G_t;g^{\alpha})$. We note that the solution of~\eqref{eq:m_FP}, $\tgamma_t$, appears both on the left of~\eqref{eq:m_FP} and on the right side in the update of $\mu^G_t$, and is thus unlike the fixed-point equation found in Bayesian Nash equilibrium.

			Furthermore, using the quantity $\tilde{\gamma}_t$ found above, define $\forall \alpha$
			\begin{align}
				V^{\alpha}_{t}(\mu^G_t,x_t^{\alpha})\defeq\E^{\tilde{\gamma}^{\alpha}_t(\cdot|x^{\alpha})}\left[R(X_t^{\alpha},A_t^{\alpha},\mu^G_t;g^{\alpha})+\delta V^{\alpha}_{t+1}(\phi(\mu^G_t,\tgamma^{\alpha}_t), X_{t+1}^{\alpha}) | \mu^G_t,x_t^{\alpha}\right].  \label{eq:Vdef}
			\end{align}
   		\end{enumerate}
		Then, an equilibrium strategy is defined as 
		\begin{align}
			\tilde{\sigma}_t^{\alpha}(a_t^{\alpha}|\mu^G_{1:t},x_{1:t}^{\alpha}) = \tilde{\gamma}^{\alpha}_t(a_t^{\alpha}|x_t^{\alpha}), \label{eq:sigma_fh}
		\end{align}
		where $\tilde{\gamma}_t = \theta[\mu^G_t]$. 

		In the following theorem, we show that the strategy thus constructed is a \glspl{gmfg} of the game.

%%%%%%%%%%%%%
		\begin{theorem}
		\label{Thm:Main}
			A strategy $(\tsigma)$ constructed from the above algorithm is an \mpe of the game i.e. $\forall t, h^{\alpha}_t \in \cH^{\alpha}_t, {\sigma^{\alpha}}$,
			\begin{align}
				\label{eq:prop}
				\E^{(\tsigma^{\alpha} \tsigma^{-\alpha})}&\left[\sum_{n=t}^T \delta^{n-t} R(X_n^{\alpha},A_n^{\alpha},\mu^G_n;g^{\alpha})|\mu^G_{1:t},x_{1:t}^{\alpha}\right] \geq \nn\\
				&\E^{({\sigma}^{\alpha} \tsigma^{-\alpha})}\left[\sum_{n=t}^T \delta^{n-t} R(X_n^{\alpha},A_n^{\alpha},\mu^G_n;g^{\alpha})|\mu^G_{1:t},x_{1:t}^{\alpha}\right] 
			\end{align}
		\end{theorem}
		\begin{proof}
			Please see Appendix~\ref{app:A}.
		\end{proof}

	\subsection{Converse}
		In the following, we show that every \gmfe can be found using the above backward recursion.

		\begin{theorem}[Converse]
		\label{thm:2}
			Let $\tsigma$ be a \gmfe of the graphon mean field game. Then there exists an equilibrium generating function $\theta$ that satisfies \eqref{eq:m_FP} in backward recursion such that  $\tsigma$ is defined using $\theta$.
		\end{theorem}
		\begin{proof}
			Please see Appendix~\ref{app:C}.
		\end{proof}

	\subsection{Backward recursive algorithm for $\mathbb{G}_{\infty}$} \label{sec:fhbr}
		In this section, we consider the infinite-horizon problem $\mathbb{G}_{\infty}$, for which we assume the reward function $R$ to be absolutely bounded.

		We define an equilibrium generating function $\theta:\cZ\to\{\cX\to\mathcal{P}(\cA) \}$, where for each $\mu^G_t $, we generate $\tgamma_t = \theta[\mu^G_t]$. In addition, we generate a reward-to-go function $V:\cZ\times\cX\to\mathbb{R}$.
		These quantities are generated through a fixed-point equation as follows.

		For all $\mu^G,$ set $\tilde{\gamma} = \theta[\mu^G]$. Then $(\tilde{\gamma},V)$ are solution of the following fixed-point equation\footnote{We discuss the existence of solution of this fixed-point equation in Section~\ref{sec:exists}}, $\forall \mu^G,x^{\alpha}\in \cX$,
		\begin{align}
			\label{eq:m_FP_ih}
			\tilde{\gamma}^{\alpha}(\cdot|x^{\alpha}) \in \arg\max_{\gamma^{\alpha}(\cdot|x^{\alpha})}&\E^{\gamma^{\alpha}(\cdot|x^{\alpha})} \left[R(x^{\alpha},A^{\alpha},\mu^G;g^{\alpha})+\delta V^{\alpha}(\phi(\mu^G,\tgamma;g^{\alpha}), X^{{\alpha}'};g^{\alpha}) | \mu^G,x^{\alpha}\right] ,\\
			V^{\alpha}(\mu^G,x^{\alpha}) =\  &\E^{\tilde{\gamma}(\cdot|x^{\alpha})} \left[ R(x^{\alpha},A^{\alpha},\mu^G;g^{\alpha})+\delta V^{\alpha} (\phi(\mu^G,\tgamma;g^{\alpha}), X^{{\alpha}'}) | \mu^G,x^{\alpha}\right]. 
		\end{align}
 		where expectation in \eqref{eq:m_FP_ih} is with respect to random variable $(A^{\alpha},X^{{\alpha},\prime})$ through the measure $\gamma(a^{\alpha}|x^{\alpha})Q^{\alpha}(x^{{\alpha}'}|x^{\alpha},a^{\alpha},\mu^G)$.

		Then an equilibrium strategy is defined as 
		\begin{align}
			\tilde{\sigma}^{\alpha}(a_t^{\alpha}|\mu^G_{1:t},x_{1:t}^{\alpha}) = \tilde{\gamma}(a_t^{\alpha}|x_t^{\alpha}), \label{eq:sigma_ih}
		\end{align}
		where $\tilde{\gamma} = \theta[\mu^G_t]$.

		The following theorem shows that the strategy thus constructed is a \gmfe of the game.

		\begin{theorem}
		\label{thih}
			A strategy $(\tsigma)$ constructed from the above algorithm is a \gmfe of the game i.e. $\forall t, h^{\alpha}_t \in \cH^{\alpha}_t, {\sigma^{\alpha}}$,

			\begin{align}
				\label{eq:prop_ih}
				\E^{(\tsigma^{\alpha} \tsigma^{-\alpha})}&\left[\sum_{n=t}^\infty \delta^{n-t} R(X_n^{\alpha},A_n^{\alpha},\mu^G_n;g^{\alpha})|\mu^G_{1:t},x_{1:t}^{\alpha}\right] \geq\nn\\
				&\E^{({\sigma}^{\alpha} \tsigma^{-\alpha})}\left[\sum_{n=t}^\infty \delta^{n-t} R(X_n^{\alpha},A_n^{\alpha},\mu^G_n;g^{\alpha})|\mu^G_{1:t},x_{1:t}^{\alpha}\right], 
			\end{align} 
		\end{theorem}
		\begin{proof}
			Please see Appendix~\ref{app:D}.
		\end{proof}

	\subsection{Converse}
		In the following, we show that every \gmfe can be found using the above backward recursion.
		\begin{theorem}[Converse]
		\label{thm:2ih}
			Let $\tsigma$ be a \gmfe the graphon mean field game. Then there exists an equilibrium generating function $\theta$
			that satisfies \eqref{eq:m_FP} in backward recursion such that  $\tsigma$ is defined using $\theta$.
		\end{theorem}
		\begin{proof}
			Please see Appendix~\ref{app:idih}.
		\end{proof}

\section{Existence}
	\label{sec:exists}
	In this section, we discuss sufficient conditions for the existence of a solution of the fixed-point equations~\eqref{eq:m_FP} and~\eqref{eq:m_FP_ih}. 
	\begin{assumption}[A1]
		The action set $\cA$ is a compact set.	
	\end{assumption}

	\begin{assumption}[A2]
			$\tilde{f}\left[x_t^\alpha,a_t^\alpha,\mu_t^G;g_\alpha\right]$ and $R(x_t^{\alpha},a_t^{\alpha},\mu^G_t;g^{\alpha})$ are Lipschitz continuous in $x_t^\alpha$ and uniformly continuous with respect to $a_t^{\alpha}$.
	\end{assumption}
	
	\begin{assumption}[A3]
			The first and second derivatives of $\tilde{f}\left[x_t^\alpha,a_t^\alpha,\mu_t^G;g_\alpha\right]$ and $R(x_t^{\alpha},a_t^{\alpha},\mu^G_t;g^{\alpha})$ with respect to $x_t^\alpha$ are continuous and bounded.
	\end{assumption}
	
	\begin{assumption}[A4]
			$\tilde{f}\left[x_t^\alpha,a_t^\alpha,\mu_t^G;g_\alpha\right]$ are Lipschitz continuous in $a_t^\alpha$ and uniformly continuous with respect to $x_t^{\alpha}$.
	\end{assumption}
	
	\begin{assumption}[A5]
		For any $v\in\mathbb{R}$, $\alpha\in\left[0,1\right]$ and any probability measure ensemble $\mu_G$, the set
		\begin{align}
			S\left(x_t^\alpha, v\right)=\arg\min_{a_t^\alpha}[v\left(\tilde{f}\left[x_t^\alpha,a_t^\alpha,\mu_t^G;g_\alpha\right]\right) 
			+ R(x_t^{\alpha},a_t^{\alpha},\mu^G_t;g^{\alpha})]
		\end{align}
			is a singleton and the resulting $a_t^\alpha$ as a function of $\left(x_t^\alpha,v\right)$ is Lipschitz continuous in $\left(x_t^\alpha,v\right)$ and uniform with respect to $\mu_t^G$ and $g_\alpha$. 
	\end{assumption}
	\begin{theorem}
		Under assumptions~(A1)-(A5), there exists a solution of the fixed-point equations~\eqref{eq:m_FP} and~\eqref{eq:m_FP_ih} for every $t$.
	\end{theorem}
	
	\proof
		Under the assumption (A1)-(A5), it has been shown in~\cite{CaHu18} that there exists a solution to the \gmfg equations. Concurrently, Theorem~\ref{thm:2} and Theorem~\ref{thm:2ih} show that all \gmfe can be found using backward recursion for the finite and infinite horizon problems. This proves that under (A1)-(A5), there exists a solution of~\eqref{eq:m_FP} and~\eqref{eq:m_FP_ih} at every $t$.
	\endproof

	\section{Methodology to compute graphon mean field team optimal policies }
		\label{sec:methodology2}
	In this section, we will provide a common agent based backward recursive dynamic programming methodology to compute optimal policies for both $\mathbb{T}_T$ and $\mathbb{T}_{\infty}$. As in Section~\ref{sec:methodology}, we will consider Markovian equilibrium strategies of player $\alpha$ which depend on the common information at time $t$, $\mu^G_{t}$, and on its current type $x_t^{\alpha}$. Equivalently, player $\alpha$ takes action of the form $A_t^{\alpha}\sim \sigma_t^{\alpha}(\cdot|\mu^G_t,x_t^{\alpha})$. As before, we first generate partial function $\gamma_t^{\alpha}:\cX\to\cP(\cA)$ as a function of $\mu^G_t$ through an equilibrium generating function $\theta_t^{\alpha}:\mu^G\to(\cX\to\cP(\cA))$ such that $\gamma_t^{\alpha} = \theta_t^{\alpha}[\mu^G_t]$. Then action $A_t^{\alpha}$ is generated by applying this prescription function $\gamma_t^{\alpha}$ on player $\alpha$'s current private information $x_t^{\alpha}$, i.e. $A_t^{\alpha}\sim \gamma_t^{\alpha}(\cdot|x_t^{\alpha})$. Thus $A_t^{\alpha}\sim \sigma_t^{\alpha}(\cdot|\mu^G_{t},x_t^{\alpha}) = \theta_t^{\alpha}[\mu^G_t](\cdot|x_t^{\alpha})$.
	
	%For a given prescription function $\gamma^{\alpha}_t = \theta^{\alpha}[\mu^G_t]$, the graphon mean-field $\mu^G_t$ evolves according to the discrete-time McKean Vlasov equation, $\forall y\in\cX$ and $\forall \alpha\in[0,1]$:
% 	\begin{align}
% 		\mu^{\alpha}_{t+1}(y) =\sum_{x\in\cX}\sum_{a\in \cA} \mu^{\alpha}_t(x)\gamma^{\alpha}_t(a|x)Q\left(y|x,a,\mu^G_t;g^{\alpha}\right), \label{eq:muG_update}
% 	\end{align}
% 	which implies
% 	\begin{align}
% 		\mu^{\alpha}_{t+1}&= \phi(\mu^{\alpha}_t,\gamma^{\alpha}_t,\mu^G_t;g^{\alpha})\\
% 		\mu_{t+1}^G &= \phi(\mu_t^G,\gamma_t;g^{\alpha}) 
% 	\end{align}

	\subsection{Backward recursive algorithm for $\mathbb{T}_{T}$} \label{sec:fhbr1}
		In this subsection, we will provide a dynamic programming methodology to generate team optimal strategies of $\mathbb{T}_{T}$ of the form described above. We define an optimal  generating function $(\theta_t)_{t\in[T]}$, where $\theta_t:\mu^G\to\{\cX\to\mathcal{P}(\cA) \}$, where for each $\mu^G_t $, we generate $\gamma_t^* = \theta_t[\mu^G_t]$. In addition, we generate a reward-to-go function $(V_t)_{t\in[T]}$, where $V_t:\mu^G_t\to\mathbb{R}$. These quantities are generated through a backward recursive optimization equation as follows.
		\begin{enumerate}
			\item Initialize $\forall \mu^G_{T+1}$,
			\begin{align}
				V_{T+1}(\mu^G_{T+1}) \defeq 0.   \label{eq:VT+1}
			\end{align}

			\item For $t = T,T-1, \ldots 1, \ \forall \mu^G_t$, let $\theta_t[\mu^G_t] $ be generated as follows. Set ${\gamma}^*_t = \theta_t[\mu^G_t]$, where ${\gamma}^*_t$ is the solution of the following optimization equation,
  			\begin{align}
				 {\gamma}^{*}_t \in  \arg\max_{\gamma_t} \E^{\gamma_t}\left[\sum_{\alpha \in[0,1]} \sum_{x_t^{\alpha}}\mu_t^{\alpha}(x_t^{\alpha}) R(x_t^{\alpha},A_t^{\alpha},\mu^G_t;g^{\alpha})+\delta V_{t+1}(\phi(\mu^G_t,\gamma_t;g^{\alpha})) | \mu^G_t\right] , \label{eq:m_FP2}
			\end{align}
 			where expectation in \eqref{eq:m_FP} is with respect to random variable $A_t^{\alpha}$ through the probability measure $\gamma^{\alpha}_t(a_t^{\alpha}|x_t^{\alpha})$. 
			Furthermore, using the quantity $\gamma^*_t$ found above, define
			\begin{align}
				V_{t}(\mu^G_t)\defeq\E^{\tilde{\gamma}^{\alpha}_t(\cdot|x^{\alpha})}\left[\sum_{\alpha \in[0,1]} \sum_{x_t^{\alpha}}\mu_t^{\alpha}(x_t^{\alpha})R(x_t^{\alpha},A_t^{\alpha},\mu^G_t;g^{\alpha})+\delta V_{t+1}(\phi(\mu^G_t,\gamma^*_t;g^{\alpha})) | \mu^G_t\right].  \label{eq:Vdef2}
			\end{align}
   		\end{enumerate}
		Then, the optimal Markovian strategy is defined as 
		\begin{align}
			{\sigma}^{*,\alpha}_t(a_t^{\alpha}|\mu^G_{1:t},x_{1:t}^{\alpha}) = {\gamma}^{*,\alpha}_t(a_t^{\alpha}|x_t^{\alpha}), \label{eq:sigma_fh}
		\end{align}
		where ${\gamma}^*_t = \theta[\mu^G_t]$. 

		In the following theorem, we show that the strategy thus constructed is an optimal Markovian strategy  of the team problem.

%%%%%%%%%%%%%
		\begin{theorem}
		\label{Thm:Main}
			A strategy $(\sigma^*)$ constructed from the above algorithm is an optimal Markovian strategy of the team problem i.e. $\forall t, \mu_{1:t}^G, {\sigma^{\alpha}}$,
			\begin{align}
				\label{eq:prop}
				\E^{\sigma^{*}}&\left[\sum_{n=t}^T \delta^{n-t}\sum_{\alpha \in[0,1]} \sum_{x_t^{\alpha}}\mu_t^{\alpha}(x_t^{\alpha}) R(x_n^{\alpha},A_n^{\alpha},\mu^G_n;g^{\alpha})|\mu^G_{1:t}\right] \geq \nn\\
				&\E^{{\sigma} }\left[\sum_{n=t}^T \delta^{n-t}\sum_{\alpha \in[0,1]} \sum_{x_t^{\alpha}}\mu_t^{\alpha}(x_t^{\alpha}) R(x_n^{\alpha},A_n^{\alpha},\mu^G_n;g^{\alpha})|\mu^G_{1:t}\right] 
			\end{align}
		\end{theorem}
		\begin{proof}
		It is easy to see that $\{\mu^G_t,\gamma_t\}_t$ is a controlled Markov process for this problem since $\mu^G_{t+1} = \phi(\mu^G_t,\gamma_t)$ and the current rewards can be written as a function of $\mu^G_t,\gamma_t$.
		Thus the result is a standard application Markov decision theory~\cite{KuVa86}.
		\end{proof}

	\
	\subsection{Backward recursive algorithm for $\mathbb{T}_{\infty}$} \label{sec:fhbr}
		In this section, we consider the infinite-horizon problem $\mathbb{T}_{\infty}$, for which we assume the reward function $R$ to be absolutely bounded.

		We define an optimal generating function $\theta:\mu^G\to\{\cX\to\mathcal{P}(\cA) \}$, where for each $\mu^G_t $, we generate $\gamma^*_t = \theta[\mu^G_t]$. In addition, we generate a reward-to-go function $V:\mu^G\to\mathbb{R}$.
		These quantities are generated through a fixed-point equation as follows.

		For all $\mu^G,$ set ${\gamma}^* = \theta[\mu^G]$. Then $({\gamma}^*,V)$ are solution of the following fixed-point equation, $\forall \mu^G$,
		\begin{align}
			\label{eq:m_FP_ih2}
			{\gamma}^{*} &\in \arg\max_{\gamma}\ \E^{\gamma} \left[\sum_{\alpha \in[0,1]} \sum_{x^{\alpha}}\mu_t^{\alpha}(x^{\alpha})R(x^{\alpha},A^{\alpha},\mu^G;g^{\alpha})+\delta V(\phi(\mu^G,\gamma^*;g^{\alpha}) | \mu^G\right] ,\\
			V(\mu^G) &=\  \E^{\tilde{\gamma}(\cdot|x^{\alpha})} \left[ \sum_{\alpha \in[0,1]} \sum_{x_t^{\alpha}}\mu_t^{\alpha}(x_t^{\alpha}) R(x^{\alpha},A^{\alpha},\mu^G;g^{\alpha})+\delta V (\phi(\mu^G,\gamma^*;g^{\alpha})) | \mu^G\right]. 
		\end{align}
 		where expectation in \eqref{eq:m_FP_ih} is with respect to random variable $(A^{\alpha})$ through the measure $\gamma(a^{\alpha}|x^{\alpha})$.

		Then the optimal Markovian strategy is defined as 
		\begin{align}
			{\sigma}_t^{*,\alpha}(a_t^{\alpha}|\mu^G_{1:t},x_{1:t}^{\alpha}) = {\gamma}_t^{*,\alpha}(a_t^{\alpha}|x_t^{\alpha}), \label{eq:sigma_ih}
		\end{align}
		where ${\gamma}_t^{*} = \phi[\mu^G_t]$.

		The following theorem shows that the strategy thus constructed is an optimal Markovian policy of the team problem.

		\begin{theorem}
		\label{thih}
			A strategy $(\sigma^*)$ constructed from the above algorithm is an optimal Markovian policy of the team problem i.e. $\forall t, \mu_{1:t}^G, {\sigma^{\alpha}}$,

			\begin{align}
				\label{eq:prop_ih2}
				\E^{\sigma^{*} }&\left[\sum_{n=t}^\infty \delta^{n-t}\sum_{\alpha \in[0,1]} \sum_{x_t^{\alpha}}\mu_t^{\alpha}(x_t^{\alpha}) R(x_n^{\alpha},A_n^{\alpha},\mu^G_n;g^{\alpha})|\mu^G_{1:t}\right] \geq\nn\\
				&\E^{{\sigma}}\left[\sum_{n=t}^\infty \delta^{n-t} \sum_{\alpha \in[0,1]} \sum_{x_t^{\alpha}}\mu_t^{\alpha}(x_t^{\alpha})R(x_n^{\alpha},A_n^{\alpha},\mu^G_n;g^{\alpha})|\mu^G_{1:t}\right], 
			\end{align} 
		\end{theorem}
		\begin{proof}
		By same argument as proof of Theorem~6, since $\{\mu^G_t,\gamma_t\}_t$ is a controlled Markov process for this problem as $\mu^G_{t+1} = \phi(\mu^G_t,\gamma_t)$ and the current rewards can be written as a function of $\mu^G_t,\gamma_t$. Also $R$ is absolutely bounded. Therefore, the result is a standard application Markov decision theory~\cite{KuVa86}.
		\end{proof}

\section{Numerical Example}
	\label{sec:example}
In this section, we put forth a numerical example to showcase the proposed sequential decomposition in the context of a system where the relative position of the players with respect to other players in a graph affects the state of the player as well as their equilibrium strategies. We provide the following definition.
\begin{definition}
	Players $\alpha$ and $\beta$ are statistically equivalent if $\phi^{\alpha}(\mu,\gamma,\mu^G;g^{\alpha}) = \phi^{\beta}(\mu,\gamma,\mu^G;g^{\beta})$.
\end{definition}

\begin{proposition}
	\label{prop:stat_equi}
	Mean field games with statistically equivalent players share the same mean field distribution and $\mu^G$ can be replaced by $\mu^\alpha$ for all $\alpha\in[0,1]$.
\end{proposition}

For a complete, Erdos R\'enyi, symmetric stochastic block model, and random geometric graphon, every player is statistically equivalent. Thus from proposition~\ref{prop:stat_equi}, the players share the same \mkv mean field evolution function and so the same mean field.

%\end{fact}
Let $n$ be the total number of statistically different players. Then $\mu^G$ can be replaced by $\{\mu\}_{i=1,\ldots, n}$. With the proposition we can represent the graphon mean field population state as
\begin{align}
	\mu^G &= \{\mu\}_{i=1,\ldots, n}=\mu\ \forall i
\end{align} 
\subsection{Cybersecurity Example}
%\begin{fact}

	We consider a cyber-security example where a cluster of nodes, facing a possible malware attack in a network, do a cost-benefit analysis to determine whether to opt for repairing. The results of this analysis, however, could be extended to many different cases like the vaccination in a population, entry and exit of firms, financial markets, demand response in smart-grid  and so on. The dynamics of each of the node is affected by the action of the neighboring nodes connected with different measures captured in a network graph and represented as a graphon function $G$. In this example, we assume different graphon functions and obtain the optimal policies using our sequential decomposition algorithm assuming that the graphs are symmetric with respect to the participating agents. In the model, the node can have two states $x^\alpha\in\cX=\{0,1\}$ representing healthy and infected node respectively. Similarly, there are two actions at their disposal for each of the state $a^\alpha\in\cA=\{0,1\}$ which says whether the nodes gets repaired with a cost or takes the risk by not undergoing repair. The chances of a node getting affected by a malware attack depends on the population as well as the state of the neighboring nodes according to the graphon. The dynamics of the model are given as
	\begin{align}
		x_{t+1}^\alpha = \lb {x_t^\alpha+\left(1-x_t^\alpha\right)w_t^\alpha &\text{   if } a_t^\alpha=0 \\
		0 &\text{ otherwise }}
	\end{align}
	where $w_t^\alpha\in\{0,1\}$ is a binary random variable with 
	\begin{align}
		\mP\{w_t^\alpha=1\}= \int_{\beta\in[0,1]}\sum_{x^\beta\in\cX}&f\left(x^\alpha_t,a^\alpha_t,x^\beta\right)g\left(\alpha,\beta\right)\mu_t^\beta(x^\beta)d(x^\beta)d\beta 
	\end{align}
	It is assumed that the value of $\mP\{w_t^\alpha=1\}$ is $q$ when the graph is fully connected i.e. $g\left(\alpha,\eta\right)=1$ and the mean state of the neighbors $\mu_t(x^\beta) = 1$. The value of $q$ is assumed to be $0.9$ for our game and $0.4$ for the team simulations. The reward function is given as
	\begin{align}
		r\left(x_t^\alpha,a_t^\alpha,\mu_t\right)=-kx_t^\alpha-\lambda a_t^\alpha
	\end{align} 
	The value $k$ represents the penalty if the node gets infected and $\lambda$ represents the cost of repair. The values $k$ and $\lambda$ are assumed as $0.3$ and $0.2$ respectively for our simulation. 
	Here we implement our algorithm to derive equilibrium for this problem by considering three popular network models to capture the interaction between the population. 
	We consider the following graphons:
	\begin{enumerate}
				\item \textit{Fully Connected Graph}: The graphon function is given as
		\begin{align}
		g(\alpha,\beta)  = 1\ \forall \alpha,\beta
		\end{align}
		
		\item \textit{Erd\"os Renyi Graph}: The graphon function is given as
			\begin{align}
				g(\alpha,\beta)  = p\ \forall \alpha,\beta
			\end{align}
We assume a value $p=0.8$ for our simulation. 
		\item \textit{Stochastic Block Model}: The graphon function is given as 
			\begin{align}
				g(\alpha,\beta) = \lb {p \text{ if } \alpha,\beta\leq 0.5 \text{ or } \alpha,\beta\geq 0.5\\
				q \text{ otherwise }}
			\end{align}	
Here, $p$ represents the intra-community interaction and is assumed as $p=.9$ for our simulation. Similarly, $q=.4$ represents the inter-community interaction parameter. 	
		\item \textit{Random Geometric graph}: The graphon function is given as
		\begin{align}
			g(\alpha,\beta) = f(\min (\beta-\alpha,1-\beta+\alpha))
		\end{align}
		where $f:[0,.5]\to [0,1]$ is a non-increasing function, and in our simulation we assume it to be $f(x) = e^{\frac{x}{0.5-x}}$.
	\end{enumerate}

%\begin{figure}[!htb]
%	\centering
%	\includegraphics[width=.6\textwidth]{Matlab_Code/Figures/Policy_0.eps}
%	\caption{Policy action at lower state for all graphons}
%\end{figure}

\begin{figure}[!htb]
\begin{center}
	\label{fig:policy}
	\centering
	\includegraphics[width=.6\textwidth]{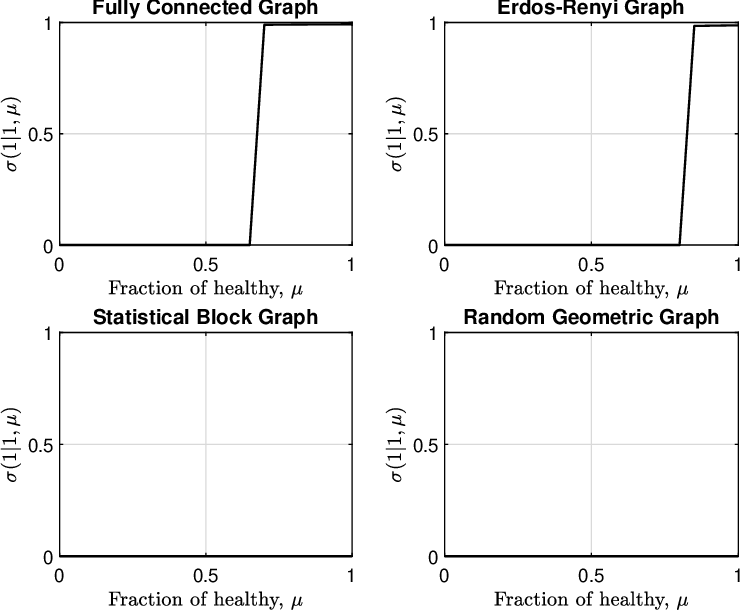}
	\caption{Policy action at higher state for all graphons for the game}
	\end{center}
\end{figure}

%\begin{figure}[!htb]
%\centering
%\includegraphics[width=.6\textwidth]{Matlab_Code/Figures/Value_0.eps}
%\caption{Value function at lower state for all graphons}
%\end{figure}
%
%\begin{figure}[!htb]
%	\centering
%	\includegraphics[width=.6\textwidth]{Matlab_Code/Figures/Value_1.eps}
%	\caption{Value function at higher state for all graphons}
%\end{figure}

\begin{figure}[!htb]
	\label{fig:Phi}
	\centering
	\includegraphics[width=.6\textwidth]{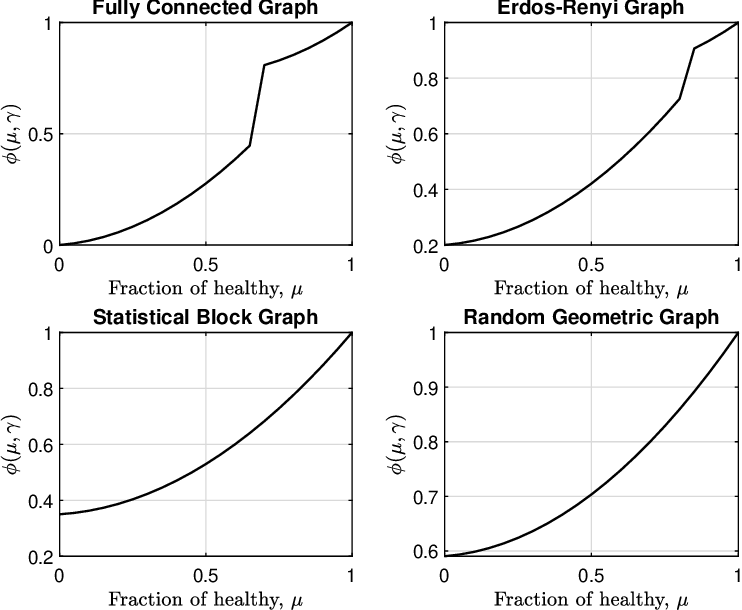}
	\centering
	\caption{Mean Field Evolution at different mean fields for the game problem}
\end{figure}

\begin{figure}[!htb]
	\label{fig:mfe}
	\centering
	\includegraphics[width=.6\textwidth]{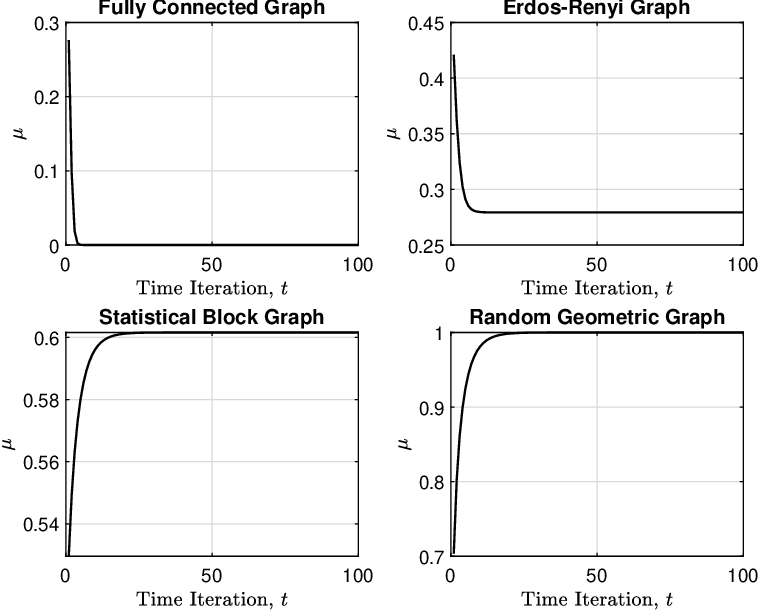}
	\caption{Convergence to \glspl{gmfg} for all graphons with time for the game problem}
\end{figure}

Figure~\ref{fig:policy} shows the equilibrium policy derived for different graphons for the specific cyber-security example. The policies differ as the interaction of the agents with their neighbors influences their strategies. Figure~\ref{fig:Phi} gives the relation between $\mu_t$ and $\mu_{t+1}$ as presented in the \eqref{eq:muG_update}. Figure~\ref{fig:mfe} shows the equilibrium mean field or in the specific case that we consider when with time, the a mean field distribution of $0.5$ approaches different mean field states for different graphons but with the same state dynamics. In Figures~\ref{fig:policy1},~\ref{fig:Phi1},~\ref{fig:mfe1}, we plot the policies and mean field equilibrium or different graphons for the specific cyber-security example when the agents cooperate as a team.

\begin{figure}[!htb]
\begin{center}
	\label{fig:policy1}
	\centering
	\includegraphics[width=.6\textwidth]{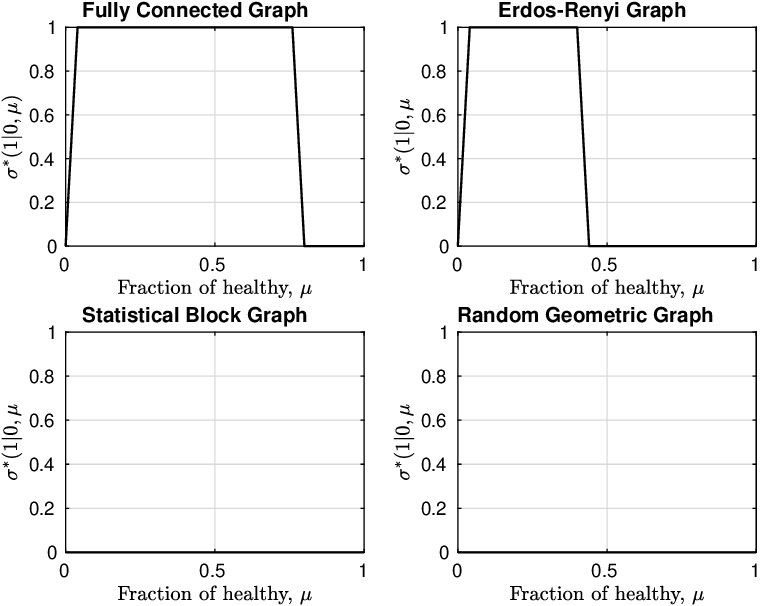}
	\caption{Policy action at higher state for all graphons for the team problem}
	\end{center}
\end{figure}
\begin{figure}[!htb]
	\label{fig:Phi1}
	\centering
	\includegraphics[width=.6\textwidth]{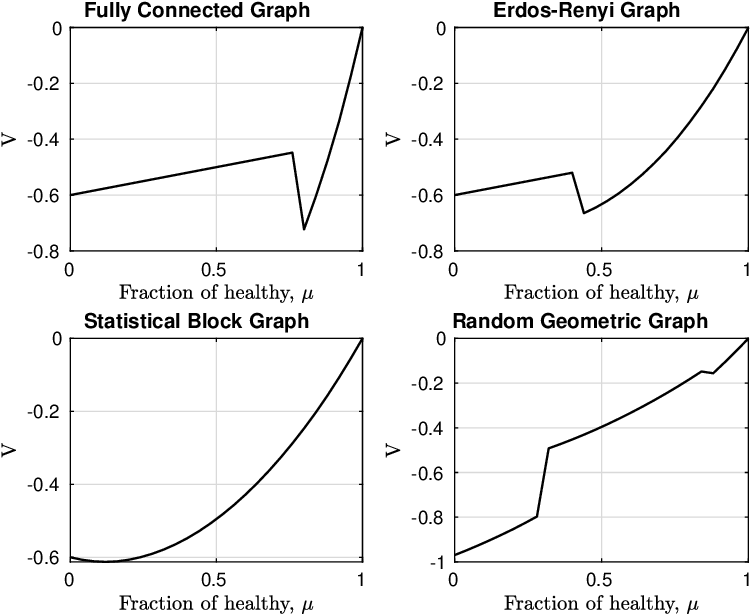}
	\centering
	\caption{Mean Field Evolution at different mean fields for the team problem}
\end{figure}

\begin{figure}[!htb]
	\label{fig:mfe1}
	\centering
	\includegraphics[width=.6\textwidth]{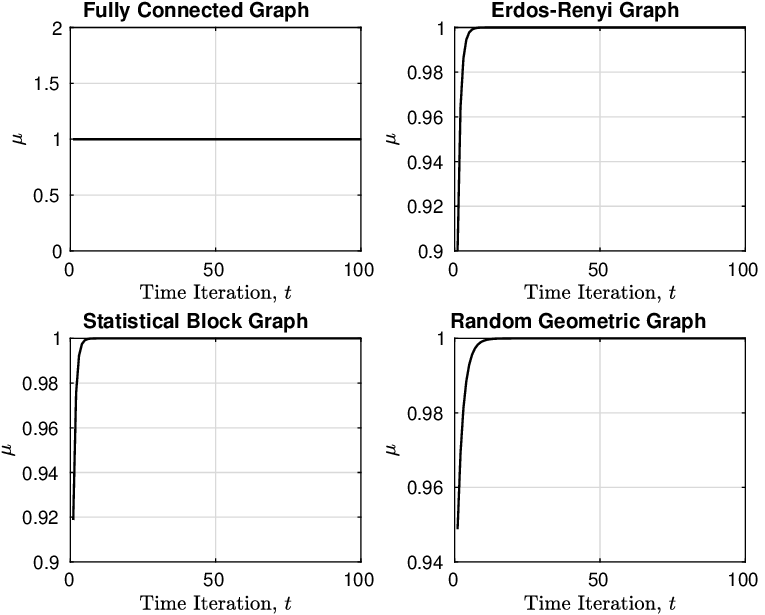}
	\caption{Convergence to \glspl{gmfg} for all graphons with time for the team problem}
\end{figure}
\section{Conclusion}
	\label{sec:Concl}
In this paper, we consider both finite and infinite horizon, large population dynamic game (with individual rewards) and team(with common rewards) where each player is affected by others through a graphon mean-field population state. We present a novel backward recursive algorithm to compute non-stationary, signaling \gmfg and GOMP for such games, where each player's strategy depends on its current private type and the current graphon mean-field population state. The non-triviality in the problem is that the update of population state is coupled to the strategies of the game, and is managed in the algorithm through unique construction of the fixed-point equations~\eqref{eq:m_FP},\eqref{eq:m_FP_ih} for GMFE and through an optimization problem~\eqref{eq:m_FP2} for the team problem. We proved the existence of the fixed-point equations~\eqref{eq:m_FP} under certain conditions. Using this algorithm, we considered a malware propagation problem where we numerically computed equilibrium and team optimal strategies of the players. In general, this algorithm be could instrumental in studying non-stationary equilibria and optimal control in a number of applications such as financial markets, social learning, renewable energy and more.

\appendices	
%\begin{APPENDICES}
    \section{}
   \label{app:A}
   \proof
   We prove~\eqref{eq:prop} using induction and the results in Lemma~\ref{lemma:2}, and \ref{lemma:1} proved in Appendix~\ref{app:B}.
   \seq{
   For base case at $t=T$, $\forall (\mu^G_{1:T}, x_{1:T}^{\alpha})\in \mathcal{H}_{T}^{\alpha}, \sigma^{\alpha}$
   \eq{
   \E^{\tsigma_{T}^{\alpha} \tsigma_{T}^{-\alpha} }\left\{  R(X_t^{\alpha},A_t^{\alpha},\mu^G_t;g^{\alpha}) \big\lvert \mu^G_{1:T}, x_{1:T}^{\alpha} \right\}
   &=
   V_T^{\alpha}(\mu^G_T, x_T^{\alpha})  \label{eq:T2a}\\
   &\geq \E^{\sigma_{T}^{\alpha} \tsigma_{T}^{-\alpha}} \left\{ R(X_t^{\alpha},A_t^{\alpha},\mu^G_t;g^{\alpha}) \big\lvert \mu^G_{1:T}, x_{1:T}^{\alpha} \right\},  \label{eq:T2}
   }
   }
   where \eqref{eq:T2a} follows from Lemma~\ref{lemma:1} and \eqref{eq:T2} follows from Lemma~\ref{lemma:2} in Appendix~\ref{app:B}.
   
   Let the induction hypothesis be that for $t+1$, $\forall i\in [N], \mu^G_{1:t+1} \in (\mathcal{H}_{t+1}^c), x_{1:t+1}^{\alpha} \in (\cX)^{t+1}, \sigma^{\alpha}$,
   \seq{
   \eq{
    \E^{\tsigma_{t+1:T}^{\alpha} \tsigma_{t+1:T}^{-\alpha}} \left\{ \sum_{n=t+1}^T \delta^{n-t-1}R(X_n^{\alpha},A_n^{\alpha},\mu^G_n;g^{\alpha}) \big\lvert \mu^G_{1:t+1}, x_{1:t+1}^{\alpha} \right\} \\
    \geq%\nn\\
     \E^{\sigma_{t+1:T}^{\alpha} \tsigma_{t+1:T}^{-\alpha}} \left\{ \sum_{n=t+1}^T \delta^{n-t-1} R(X_n^{\alpha},A_n^{\alpha},\mu^G_n;g^{\alpha}) \big\lvert  \mu^G_{1:t+1}, x_{1:t+1}^{\alpha} \right\}. \label{eq:PropIndHyp}
   }
   }
   \seq{
   Then $\forall i\in [N], (\mu^G_{1:t}, x_{1:t}^{\alpha}) \in \mathcal{H}_{t}^{\alpha}, \sigma^{\alpha}$, we have
   \eq{
   &\E^{\tsigma_{t:T}^{\alpha} \tsigma_{t:T}^{-\alpha}} \left\{ \sum_{n=t}^T \delta^{n-t-1}R(X_n^{\alpha},A_n^{\alpha},\mu^G_n;g^{\alpha}) \big\lvert \mu^G_{1:t}, x_{1:t}^{\alpha} \right\} \nonumber \\
   &= V_t(\mu^G_{t}, x_t^{\alpha})\label{eq:T1}\\
   &\geq \E^{\sigma_t^{\alpha} \tsigma_t^{-\alpha}} \left\{ R(X_t^{\alpha},A_t^{\alpha},\mu^G_t;g^{\alpha}) + \delta V_{t+1}^{\alpha} (\mu^G_{t+1}, X_{t+1}^{\alpha}) \big\lvert \mu^G_{1:t}, x_{1:t}^{\alpha} \right\}  \label{eq:T3}\\
   &= \E^{\sigma_t^{\alpha} \tsigma_t^{-\alpha}} \left\{ R(X_t^{\alpha},A_t^{\alpha},\mu^G_t;g^{\alpha}) + \right.\nn\\
   &\left.\delta \E^{\tsigma_{t+1:T}^{\alpha} \tsigma_{t+1:T}^{-\alpha}} \left\{ \sum_{n=t+1}^T \delta^{n-t-1}R(X_n^{\alpha},A_n^{\alpha},\mu^G_n;g^{\alpha}) \big\lvert  \mu^G_{1:t},\mu^G_{t+1}, x_{1:t}^{\alpha},X_{t+1}^{\alpha} \right\}  \big\vert \mu^G_{1:t}, x_{1:t}^{\alpha} \right\}  \label{eq:T3b}\\
   &\geq \E^{\sigma_t^{\alpha} \tsigma_t^{-\alpha}} \left\{ R(X_t^{\alpha},A_t^{\alpha},\mu^G_t;g^{\alpha}) + \right.\nn\\
   &\left.\delta\E^{\sigma_{t+1:T}^{\alpha} \tsigma_{t+1:T}^{-\alpha} } \left\{ \sum_{n=t+1}^T \delta^{n-t-1}R(X_n^{\alpha},A_n^{\alpha},\mu^G_n;g^{\alpha}) \big\lvert \mu^G_{1:t},\mu^G_{t+1}, x_{1:t}^{\alpha},X_{t+1}^{\alpha}\right\} \big\vert \mu^G_{1:t}, x_{1:t}^{\alpha} \right\}  \label{eq:T4} \\
   &= \E^{\sigma_t^{\alpha} \tsigma_t^{-\alpha}} \big\{ R(X_t^{\alpha},A_t^{\alpha},\mu^G_t;g^{\alpha}) + \nn\\
   &\delta\E^{\sigma_{t:T}^{\alpha} \tsigma_{t:T}^{-\alpha} } 
   \left\{ \sum_{n=t+1}^T \delta^{n-t-1}R(X_n^{\alpha},A_n^{\alpha},\mu^G_n;g^{\alpha}) \big\lvert \mu^G_{1:t},\mu^G_{t+1}, x_{1:t}^{\alpha},X_{t+1}^{\alpha}\right\} \big\vert \mu^G_{1:t}, x_{1:t}^{\alpha} \big\}  
   \label{eq:T5}\\
   &=\E^{\sigma_{t:T}^{\alpha} \tsigma_{t:T}^{-\alpha}} \left\{ \sum_{n=t}^T \delta^{n-t}R(X_n^{\alpha},A_n^{\alpha},\mu^G_n;g^{\alpha}) \big\lvert \mu^G_{1:t},  x_{1:t}^{\alpha} \right\}  \label{eq:T6},
   }
   }
   where \eqref{eq:T1} follows from Lemma~\ref{lemma:1}, \eqref{eq:T3} follows from Lemma~\ref{lemma:2}, \eqref{eq:T3b} follows from Lemma~\ref{lemma:1}, \eqref{eq:T4} follows from induction hypothesis in \eqref{eq:PropIndHyp} and \eqref{eq:T5} follows since the random variables involved in the right conditional expectation do not depend on strategies $\sigma_t^{\alpha}$.
   \endproof
   
   \section{}
   \label{app:B}
   \begin{lemma}
   \label{lemma:2}
   $\forall t\in [T], i\in [N], (\mu^G_{1:t}, x_{1:t}^{\alpha})\in \mathcal{H}_t^{\alpha}, \sigma^{\alpha}_t$
   \eq{
   V_t^{\alpha}(\mu^G_t, x_t^{\alpha}) \geq \E^{\sigma_t^{\alpha} \tsigma_t^{-\alpha}} \left\{ R(X_t^{\alpha},A_t^{\alpha},\mu^G_t;g^{\alpha}) + \delta V^{\alpha}_{t+1} (\mu^G_{t+1}, X_{t+1}^{\alpha}) \big\lvert  \mu^G_{1:t}, x_{1:t}^{\alpha} \right\}.\label{eq:lemma2}
   }
   \end{lemma}
   
   \proof
   We prove this lemma by contradiction.
   
    Suppose the claim is not true for $t$. This implies $\exists i, \hat{\sigma}_t^{\alpha}, \hat{\mu}^G_{1:t}, \hat{x}_{1:t}^{\alpha}$ such that
   \eq{
   \E^{\hat{\sigma}_t^{\alpha} \tsigma_t^{-\alpha}} \left\{ R(X_t^{\alpha},A_t^{\alpha},\mu^G_t;g^{\alpha}) +  \delta V^{\alpha}_{t+1} (\mu^G_{t+1}, X_{t+1}^{\alpha}) \big\lvert \hat{\mu}^G_{1:t},\hat{x}_{1:t}^{\alpha} \right\} 
   > V_t(\hmu^G_t, \hat{x}_{t}^{\alpha}).\label{eq:E8}
   }
   We will show that this leads to a contradiction.
   Construct 
   \begin{equation}
   \hat{\gamma}^{\alpha}_t(a_t^{\alpha}|x_t^{\alpha}) = \lb{\hat{\sigma}_t^{\alpha}(a_t^{\alpha}|\hat{\mu}^G_{1:t},\hat{x}_{1:t}^{\alpha}) \;\;\;\;\; x_t^{\alpha} = \hat{x}_t^{\alpha} \\ \text{arbitrary} \;\;\;\;\;\;\;\;\;\;\;\;\;\; \text{otherwise.}  }
   \end{equation}
   
   Then for $\hat{\mu}^G_{1:t}, \hat{x}_{1:t}^{\alpha}$, we have
   \seq{
   \eq{
   &V^{\alpha}_t(\hmu^G_t, \hat{x}_t^{\alpha}) 
   \nn \\
   &= \max_{\gamma_t(\cdot|\hat{x}_t^{\alpha})} \E^{\gamma_t(\cdot|\hat{x}_t^{\alpha}) \tsigma_t^{-\alpha}} \left\{ R(\hat{x}_t^{\alpha},A_t^{\alpha},\hmu^G_t) + \delta V^{\alpha}_{t+1} (\phi(\hmu^G_t,\tgamma_t;g^{\alpha}), X_{t+1}^{\alpha}) \big\lvert \hmu^G_t, \hat{x}_{t}^{\alpha} \right\}, \label{eq:E11}\\
   &\geq\E^{\hat{\gamma}_t^{\alpha}(\cdot|\hat{x}_t^{\alpha}) \tsigma_t^{-\alpha}} \left\{ R(X_t^{\alpha},A_t^{\alpha},\mu^G_t;g^{\alpha}) + \delta V^{\alpha}_{t+1} (\phi(\hmu^G_t,\tgamma_t;g^{\alpha}), {X}_{t+1}^{\alpha}) \big\lvert \hmu^G_t,\hat{x}_{t}^{\alpha} \right\}   
   \\ \nn 
   &=\sum_{a_t^{\alpha},x_{t+1}^{\alpha}}   \left\{ R(\hat{x}_t^{\alpha},a_t^{\alpha},\hmu^G_t) + \delta V_{t+1} (\phi(\hmu^G_t,\tgamma_t;g^{\alpha}), x_{t+1}^{\alpha})\right\}
   \hat{\gamma}_t(a^{\alpha}_t|\hat{x}_t^{\alpha})Q^{\alpha}(x_{t+1}^{\alpha}|\hat{x}_t^{\alpha},a_t^{\alpha},\hmu^G_t;g^{\alpha})  
   \\ \nn 
   &= \sum_{a_t^{\alpha},x_{t+1}^{\alpha}}  \left\{ R(\hat{x}_t^{\alpha},a_t^{\alpha},\hmu^G_t) + \delta V^{\alpha}_{t+1} (\phi(\hmu^G_t,\tgamma_t;g^{\alpha}), x_{t+1}^{\alpha})\right\}
   \hat{\sigma}_t(a_t^{\alpha}|\hat{\mu}^G_{1:t} ,\hat{x}_{1:t}^{\alpha}) Q^{\alpha}(x_{t+1}^{\alpha}|\hat{x}_t^{\alpha},a_t^{\alpha},\hmu^G_t;g^{\alpha}) \label{eq:E9}\\
   &= \E^{\hat{\sigma}_t^{\alpha} \tsigma_t^{-\alpha}} \left\{ R(\hat{x}_t^{\alpha},a_t^{\alpha},\hmu^G_t)+ \delta V^{\alpha}_{t+1} (\phi(\hmu^G_t,\tgamma_t), X_{t+1}^{\alpha};g^{\alpha}) \big\lvert \hat{\mu}^G_{1:t},  \hat{x}_{1:t}^{\alpha} \right\}  \\
   &> V_t^{\alpha}(\hmu^G_t, \hat{x}_{t}^{\alpha}), \label{eq:E10}
   }
   where \eqref{eq:E11} follows from definition of $V_t$ in \eqref{eq:Vdef}, \eqref{eq:E9} follows from definition of $\hat{\gamma}_t^{\alpha}$ and \eqref{eq:E10} follows from \eqref{eq:E8}. However this leads to a contradiction.
   }
   \endproof

   \begin{lemma}
   \label{lemma:1}
   $\forall i\in [N], t\in [T], (\mu^G_{1:t}, x_{1:t}^{\alpha})\in \mathcal{H}_t^{\alpha}$,
   \begin{gather}
   V^{\alpha}_t(\mu^G_{t}, x_t^{\alpha}) =
   %\\ \hs{-0.2cm} 
   \E^{\tsigma_{t:T}^{\alpha} \tsigma_{t:T}^{-\alpha}} \left\{ \sum_{n=t}^T \delta^{n-t}R(X_n^{\alpha},A_n^{\alpha},\mu^G_n;g^{\alpha}) \big\lvert  \mu^G_{1:t}, x_{1:t}^{\alpha} \right\} .
   \end{gather} 
   \end{lemma}
   
   \proof
   %This lemma is in the same spirit as the following statement: ``For a controlled Markov process, if Markov policies are played, then the resulting process is a Markov process, where reward-to-go at any time can be denoted by a function of the current state."
   %
   \seq{
   We prove the lemma by induction. For $t=T$,
   \eq{
    \E^{\tsigma_{T}^{\alpha} \tsigma_{T}^{-\alpha} } \left\{  R(X_t^{\alpha},A_t^{\alpha},\mu^G_t;g^{\alpha}) \big\vert V^{\alpha}ert \mu^G_{1:T},  x_{1:T}^{\alpha} \right\}
    &= \sum_{a_T^{\alpha}} R(X_t^{\alpha},A_t^{\alpha},\mu^G_t;g^{\alpha}) \tsigma_{T}^{\alpha}(a_T^{\alpha}|\mu^G_{T},x_{T}^{\alpha}) \\
    &= V^{\alpha}_T(\mu^G_{T}, x_T^{\alpha}) \label{eq:C1},
   }
   }
   where \eqref{eq:C1} follows from the definition of $V^{\alpha}_t$ in \eqref{eq:Vdef}.
   Suppose the claim is true for $t+1$, i.e., $\forall i\in [N], t\in [T], (\mu^G_{1:t+1}, x_{1:t+1}^{\alpha})\in \mathcal{H}_{t+1}^{\alpha}$
   \begin{gather}
   V^{\alpha}_{t+1}(\mu^G_{t+1}, x_{t+1}^{\alpha}) = \E^{\tsigma_{t+1:T}^{\alpha} \sigma_{t+1:T}^{-\alpha}}
   %\\
   \left\{ \sum_{n=t+1}^T \delta^{n-t-1}R(X_n^{\alpha},A_n^{\alpha},\mu^G_n;g^{\alpha}) \big\lvert \mu^G_{1:t+1}, x_{1:t+1}^{\alpha} \right\} 
   \label{eq:CIndHyp}.
   \end{gather}
   Then $\forall i\in [N], t\in [T], (\mu^G_{1:t}, x_{1:t}^{\alpha})\in \mathcal{H}_t^{\alpha}$, we have
   \seq{
   \eq{
   &\E^{\tsigma_{t:T}^{\alpha} \tsigma_{t:T}^{-\alpha} } \left\{ \sum_{n=t}^T \delta^{n-t} R(X_n^{\alpha},A_n^{\alpha},\mu^G_n;g^{\alpha}) \big\lvert  \mu^G_{1:t}, x_{1:t}^{\alpha} \right\} 
   \nonumber 
   \\
   &=  \E^{\tsigma_{t:T}^{\alpha} \tsigma_{t:T}^{-\alpha} } \left\{R(X_t^{\alpha},A_t^{\alpha},\mu^G_t;g^{\alpha})  \right.
   \nonumber \\ 
   &\left.+\delta \E^{\tsigma_{t:T}^{\alpha} \tsigma_{t:T}^{-\alpha} }  \left\{ \sum_{n=t+1}^T \delta^{n-t-1}R(X_n^{\alpha},A_n^{\alpha},\mu^G_n;g^{\alpha})\big\lvert \mu^G_{1:t},  \mu^G_{t+1}, x_{1:t}^{\alpha},X_{t+1}^{\alpha}\right\} \big\lvert \mu^G_{1:t},  x_{1:t}^{\alpha} \right\} \label{eq:C2}
   \\
   &=  \E^{\tsigma_{t:T}^{\alpha} \tsigma_{t:T}^{-\alpha} } \left\{R(X_t^{\alpha},A_t^{\alpha},\mu^G_t;g^{\alpha}) \right.
   \nonumber 
   \\
   &\left.  +\delta\E^{\tsigma_{t+1:T}^{\alpha} \tsigma_{t+1:T}^{-\alpha} }\left\{ \sum_{n=t+1}^T \delta^{n-t-1}R(X_n^{\alpha},A_n^{\alpha},\mu^G_n;g^{\alpha})\big\lvert \mu^G_{1:t},\mu^G_{t+1}, x_{1:t}^{\alpha},X_{t+1}^{\alpha}\right\} \big\lvert \mu^G_{1:t}, x_{1:t}^{\alpha} \right\} \label{eq:C3}
   \\
   &=  \E^{\tsigma_{t:T}^{\alpha} \tsigma_{t:T}^{-\alpha} } \left\{R(X_t^{\alpha},A_t^{\alpha},\mu^G_t;g^{\alpha}) +  \delta V^{\alpha}_{t+1}(\mu^G_{t+1}, X_{t+1}^{\alpha}) \big\lvert  \mu^G_{1:t}, x_{1:t}^{\alpha} \right\} 
   \label{eq:C4}
   \\
   &=  \E^{\tsigma_{T}^{\alpha} \tsigma_{T}^{-\alpha}} \left\{R(X_t^{\alpha},A_t^{\alpha},\mu^G_t;g^{\alpha}) +  \delta V^{\alpha}_{t+1}(\mu^G_{t+1}, X_{t+1}^{\alpha}) \big\lvert  \mu^G_{1:t}, x_{1:t}^{\alpha} \right\} 
   \label{eq:C5}
   \\
   &=V^{\alpha}_{t}(\mu^G_t, x_t^{\alpha}) \label{eq:C6},
   }
   }
   \eqref{eq:C4} follows from the induction hypothesis in \eqref{eq:CIndHyp}, \eqref{eq:C5} follows because the random variables involved in expectation, $X_t^{\alpha},A_t^{\alpha},\mu^G_t,\mu^G_{t+1},X_{t+1}^{\alpha}$ do not depend on $\tsigma_{t+1:T}^{\alpha} \sigma_{t+1:T}^{-\alpha}$ and \eqref{eq:C6} follows from the definition of $V^{\alpha}_t$ in \eqref{eq:Vdef}.
   \endproof
   
   \section{}
   \label{app:C}
   \proof
   We prove this by contradiction. Suppose for any equilibrium generating function $\theta$ that generates an MPE $\tsigma$, there exists $t\in[T], i\in[N], \mu^G_{1:t}\in\cH_t^c,$ such that \eqref{eq:m_FP} is not satisfied for $\theta$
   %\footnote{Note that for $\mu^G_t \neq \mu^G_t $ for any $a_{1:t-1}$, $\phi$ can be arbitrarily defined without affecting the definition of $(\tsigma,\mu^*)$.}
   i.e. for $\tgamma_t = \theta_t[\mu^G_t] = \tsigma_t(\cdot|\mu^G_t,\cdot)$,
   \eq{
    \tilde{\gamma}_t(\cdot|x_t^{\alpha}) \not\in \arg\max_{\gamma_t(\cdot|x_t^{\alpha})} \E^{\gamma_t(\cdot|x_t^{\alpha})} \left\{ R_t(X_t^{\alpha},A_t^{\alpha},\mu^G_t;g^{\alpha}) + V^{\alpha}_{t+1}(\phi(\mu^G_t,\tilde{\gamma}_t;g^{\alpha}), X_{t+1}^{\alpha}) \big\lvert  x_t^{\alpha},\mu^G_t \right\} . \label{eq:FP4}
     }
     Let $t$ be the first instance in the backward recursion when this happens. This implies $\exists\ \hat{\gamma}_t$ such that
     \eq{
     \E^{\hat{\gamma}_t(\cdot|x_t^{\alpha})} \left\{ R_t(X_t^{\alpha},A_t^{\alpha},\mu^G_t;g^{\alpha})+ V_{t+1} (\phi(\mu^G_t, \tilde{\gamma}_t;g^{\alpha}), X_{t+1}^{\alpha}) \big\lvert  \mu^G_{1:t},x_{1:t}^{\alpha}\right\}
     \nn\\
     > \E^{\tgamma_t(\cdot|x_t^{\alpha})} \left\{ R_t(X_t^{\alpha},A_t^{\alpha},\mu^G_t;g^{\alpha}) + V^{\alpha}_{t+1} (\phi(\mu^G_t, \tilde{\gamma}_t;g^{\alpha}), X_{t+1}^{\alpha}) \big\lvert  \mu^G_{1:t},x_{1:t}^{\alpha} \right\} \label{eq:E1}
     }
     This implies for $\hat{\sigma}_t(\cdot|\mu^G_t,\cdot) = \hat{\gamma}_t$,
     \eq{
     &\E^{\tsigma_{t:T}} \left\{ \sum_{n=t}^T R_n(X_n^{\alpha},A_n^{\alpha},\mu^G_n;g^{\alpha}) \big\lvert  \mu^G_{1:t-1}, x_{1:t}^{\alpha} \right\}
     \nn\\
     %&= \E^{\tsigma_t} \left\{ R_t(X_t,A_t,\mu^G_t) + \E^{\tsigma_{t:T}}  \left\{ \sum_{n=t+1}^T R_n(X_n,A_n^{\alpha},\mu^G_n) \big\lvert \mu^G_{1:t},\mu^G_{t+1}, x_{1:t},X_{t+1} \right\}  \big\vert \mu^G_{1:t}, x_{1:t} \right\}% \label{eq:E2a}
   \\
     &= \E^{\tsigma_t^{\alpha},\tsigma_t^{-\alpha}} \left\{ R_t(X_t^{\alpha},A_t^{\alpha},\mu^G_t;g^{\alpha}) + \E^{\tsigma_{t+1:T}^{\alpha} \tsigma_{t+1:T}^{-\alpha}}   \left\{ \sum_{n=t+1}^T R_n(X_n^{\alpha},A_n^{\alpha},\mu^G_n;g^{\alpha}) \big\lvert \mu^G_{1:t-1},\mu^G_{t+1}, x_{1:t}^{\alpha},X_{t+1}^{\alpha} \right\}  \big\vert \mu^G_{1:t}, x_{1:t}^{\alpha} \right\} \label{eq:E2}
     \\
     &=\E^{\tgamma_t(\cdot|x_t) \tilde{\gamma}^{-\alpha}_t} \left\{ R_t(X_t^{\alpha},A_t^{\alpha},\mu^G_t;g^{\alpha}) + V^{\alpha}_{t+1} (\phi(\mu^G_t, \tilde{\gamma}_t;g^{\alpha}), X_{t+1}^{\alpha}) \big\lvert  \mu^G_t,x_t^{\alpha} \right\} \label{eq:E3}
     \\
     &< \E^{\hat{\sigma}_t(\cdot|\mu^G_t,x_t^{\alpha}) \tilde{\gamma}^{-\alpha}_t} \left\{ R_t(X_t^{\alpha},A_t^{\alpha},\mu^G_t;g^{\alpha}) + V^{\alpha}_{t+1} (\phi(\mu^G_t, \tilde{\gamma}_t;g^{\alpha}), X_{t+1}^{\alpha}) \big\lvert  \mu^G_t,x_t \right\}\label{eq:E4}
     \\
     &= \E^{\hat{\sigma}_t \tsigma_t^{-\alpha}} \left\{ R_t(X_t^{\alpha},A_t^{\alpha},\mu^G_t;g^{\alpha}) +  \E^{\tsigma_{t+1:T}^{\alpha} \tsigma_{t+1:T}^{-\alpha}}\left\{ \sum_{n=t+1}^T R_n(X_n^{\alpha},A_n^{\alpha},\mu^G_n;g^{\alpha}) \big\lvert \mu^G_{1:t},\mu^G_{t+1}, x_{1:t}^{\alpha},X_{t+1}^{\alpha}\right\} \big\vert \mu^G_{1:t}, x_{1:t} ^{\alpha}\right\}\label{eq:E5}
     \\
     &=\E^{\hat{\sigma}_t,\tsigma_{t+1:T}^{\alpha} \tsigma_{t:T}^{-\alpha}} \left\{ \sum_{n=t}^T R_n(X_n^{\alpha},A_n^{\alpha},\mu^G_n;g^{\alpha}) \big\lvert  \mu^G_{1:t}, x_{1:t}^{\alpha}\right\},\label{eq:E6}
     }
     where \eqref{eq:E3} follows from the definitions of $\tgamma_t$ and Lemma~\ref{lemma:1}, \eqref{eq:E4} follows from \eqref{eq:E1} and the definition of $\hat{\sigma}_t$, \eqref{eq:E5} follows from Lemma~\ref{lemma:2}. However, this leads to a contradiction since $\tsigma$ is an MPE of the game.
   \endproof
   
   \section{}
   \label{app:D}
       We divide the proof into two parts: first we show that the value function $ V $ is at least as big as any reward-to-go function; secondly we show that under the strategy $ \tsigma $, reward-to-go is $ V^{\alpha} $. Note that $h_t^{\alpha} := (\mu^G_{1:t},x_{1:t}^{\alpha})$.
   
   \paragraph*{Part 1}
   For any $ i \in \mN $, $ \sigma^{\alpha} $ define the following reward-to-go functions
   \begin{subequations} \label{eqihr2g}
   %	\begin{align} 				
   \eq{
       W_t^{\sigma^{\alpha}}(h_t^{\alpha}) &= \mE^{\sigma^{\alpha},\tsigma^{-\alpha}} \lpr \sum_{n=t}^\infty \delta^{n-t} R(X_n^{\alpha},A_n^{\alpha},\mu^G_n;g^{\alpha}) \mid h_t^{\alpha} \rpr\\
   %	\\
       %\label{eqr2gfh}		
       W_t^{\sigma^{\alpha},T}(h_t^{\alpha}) &= \mE^{\sigma^{\alpha},\tsigma^{-\alpha}} \lpr \sum_{n=t}^T \delta^{n-t} R(X_n^{\alpha},A_n^{\alpha},\mu^G_n;g^{\alpha})
   %	\\
       +  \delta^{T+1-t} V^{\alpha}(\mu^G_{T+1},X^{\alpha}_{T+1}) \mid h_t^{\alpha} \rpr.
   }
   %	\end{align}
   \end{subequations}
   Since $ \mX,\mA $ are finite sets the reward $ R$ is absolutely bounded, the reward-to-go $ W_t^{\sigma^{\alpha}}(h_t^{\alpha}) $ is finite $ \forall $ $ i,t,\sigma^{\alpha},h_t^{\alpha} $.
   
   For any $ i \in \mN $, $ h_t^{\alpha} \in \mathcal{H}_t^{\alpha} $,
   \eq{ \label{eqdc}
   %\begin{equation}
   V^{\alpha}\big(\mu^G_t,x_t^{\alpha}\big) - W_t^{\sigma^{\alpha}}(h_t^{\alpha})
   = \Big[ V^{\alpha}\big(\mu^G_t,x_t^{\alpha}\big) - W_t^{\sigma^{\alpha},T}(h_t^{\alpha}) \Big]
   %	\\
   + \Big[ W_t^{\sigma^{\alpha},T}(h_t^{\alpha}) - W_t^{\sigma^{\alpha}}(h_t^{\alpha}) \Big]
   %\end{equation} 	
       }
   Combining results from Lemmas~\ref{thmfh2} and~\ref{lemfhtoih} in Appendix~\ref{app:D}, %
   %8 and 9 in Appendix G of the technical report~\cite{VaSiAn17arxiv}, 
   the term in the first bracket in RHS of~\eqref{eqdc} is non-negative. Using~\eqref{eqihr2g}, the term in the second bracket is
   %	\eq{
   \begin{gather}\label{eqdiff}	
   \left( \delta^{T+1-t} \right) \mE^{\sigma^{\alpha},\tsigma^{-\alpha}} \Big\{- \sum_{n=T+1}^\infty \delta^{n-(T+1)} R(X_n^{\alpha},A_n^{\alpha},\mu^G_n;g^{\alpha})
   %	\\
   + V^{\alpha}(\mu^G_{T+1},X^{\alpha}_{T+1}) \mid h_t^{\alpha} \Big\}.
   \end{gather} 	
   %	}
   The summation in the expression above is bounded by a convergent geometric series. Also, $ V^{\alpha} $ is bounded. Hence the above quantity can be made arbitrarily small by choosing $ T $ appropriately large. Since the LHS of~\eqref{eqdc} does not depend on $ T $, which implies,
   \begin{gather}
   V^{\alpha}\big(\mu^G_t,x_t^{\alpha}\big) \ge W_t^{\sigma^{\alpha}}(h_t^{\alpha}).
   \end{gather}

   \paragraph*{Part 2}
   Since the strategy the equilibrium strategy $ \tsigma $ generated in~\eqref{eq:sigma_ih} is such that $\tsigma^{\alpha}_t $ depends on $ h_t^{\alpha} $ only through $ \mu^G_t $ and $ x_t^{\alpha} $, the reward-to-go $ W_t^{\tsigma^{\alpha}} $, at strategy $ \tsigma $, can be written (with abuse of notation) as
   \begin{gather}
   %\begin{gather}
   W_t^{\tsigma^{\alpha}}(h_t^{\alpha}) = W_t^{\tsigma^{\alpha}}(\mu^G_t,x_t^{\alpha})
   %\\
   = \mE^{\tsigma} \lpr \sum_{n=t}^\infty \delta^{n-t} R(X_n^{\alpha},A_n^{\alpha},\mu^G_n;g^{\alpha}) \mid \mu^G_t,x_t^{\alpha} \rpr.
   %\end{gather}
   \end{gather}
   
   For any $ h_t^{\alpha} \in \mathcal{H}_t^{\alpha} $,
   \begin{subequations}
   \eq{
       W_t^{\tsigma^{\alpha}}(\mu^G_t,x_t^{\alpha})
       &= \mE^{\tsigma} \lpr R(X_t^{\alpha},A_t^{\alpha},\mu^G_t;g^{\alpha})	+ \delta W_{t+1}^{\tsigma^{\alpha}}
   %	\\
       \big(\phi(\mu^G_t,\theta[\mu^G_t],g^{\alpha})),X_{t+1}^{\alpha}\big)  \mid \mu^G_t,x_t^{\alpha} \rpr\\
   %\\
       V^{\alpha}(\mu^G_t,x_t^{\alpha})
       &= \mE^{\tsigma} \Big\{ R(X_t^{\alpha},A_t^{\alpha},\mu^G_t;g^{\alpha}) + \delta V^{\alpha}
   %	\\
       \big(\phi(\mu^G_t,\theta[\mu^G_t],g^{\alpha})),X_{t+1}^{\alpha}\big)  \mid \mu^G_t,x_t^{\alpha} \Big\}.
   %\end{gather}
   }
   \end{subequations}
   Repeated application of the above for the first $ n $ time periods gives
   \begin{subequations}
   \eq{
   %\begin{align}
       W_t^{\tsigma^{\alpha}}(\mu^G_t,x_t^{\alpha})
       &= \mE^{\tsigma}\Bigg\{ \sum_{m=t}^{t+n-1} \delta^{m-t} R(X_t^{\alpha},A_t^{\alpha},\mu^G_t;g^{\alpha})
   %	\\
       + \delta^{n}  W_{t+n}^{\tsigma^{\alpha}}\big(\mu^G_{t+n},X_{t+n}^{\alpha}\big)  \mid \mu^G_t,x_t^{\alpha} \Bigg\}
   \\
       V^{\alpha}(\mu^G_t,x_t^{\alpha})
       &= \mE^{\tsigma} \Bigg\{ \sum_{m=t}^{t+n-1} \delta^{m-t} R(X_t^{\alpha},A_t^{\alpha},\mu^G_t;g^{\alpha})
   %	\\
       + \delta^{n}  V^{\alpha}\big(\mu^G_{t+n},X_{t+n}^{\alpha}\big)  \mid \mu^G_t,x_t^{\alpha} \Bigg\}.
   %\end{align}
   }
   \end{subequations}
   Taking differences results in
   %\begin{subequations}
       \eq{
       W_t^{\tsigma^{\alpha}}(\mu^G_t,x_t^{\alpha})  - V^{\alpha}(\mu^G_t,x_t^{\alpha})
       =\delta^n \mE^{\tsigma}
   %	\\
         \lpr W_{t+n}^{\tsigma^{\alpha}}\big(\mu^G_{t+n},X_{t+n}^{\alpha}\big)
       - V\big(\mu^G_{t+n},X_{t+n}^{\alpha}\big) \mid \mu^G_t,x_t^{\alpha} \rpr.
       }
   %\end{subequations}
   Taking absolute value of both sides then using Jensen's inequality for $ f(x) = \vert x \vert $ and finally taking supremum over $ h_t^{\alpha} $ reduces to
   \eq{
   \sup_{h_t^{\alpha}} \big\vert W_t^{\tsigma^{\alpha}}(\mu^G_t,x_t^{\alpha})  - V^{\alpha}(\mu^G_t,x_t^{\alpha}) \big\vert 
   \le \delta^n \sup_{h_t^{\alpha}}  \mE^{\tsigma}
   %\\
    \lpr\big\vert W_{t+n}^{\tsigma^{\alpha}}(\mu^G_{t+n},X_{t+n}^{\alpha})
    - V^{\alpha}(\mu^G_{t+n},X_{t+n}^{\alpha}) \big\vert  \mid \mu^G_t,x_t^{\alpha} \rpr.
   }
   Now using the fact that $ W_{t+n},V$ are bounded and that we can choose $ n $ arbitrarily large, we get $ \sup_{h_t^{\alpha}} \vert W_t^{\tsigma^{\alpha}}(\mu^G_t,x_t^{\alpha})  - V^{\alpha}(\mu^G_t,x_t^{\alpha}) \vert = 0 $. 	
   
   \section{}
   \label{app:E}
   In this section, we present three lemmas. Lemma~\ref{thmfh1} is intermediate technical results needed in the proof of Lemma~\ref{thmfh2}. Then the results in Lemma~\ref{thmfh2} and~\ref{lemfhtoih} are used in Appendix~\ref{app:C} for the proof of Theorem~\ref{thih}. The proof for Lemma~\ref{thmfh1} below isn't stated as it analogous to the proof of Lemma~\ref{lemma:2} from Appendix~\ref{app:B}, used in the proof of Theorem~\ref{Thm:Main} (the only difference being a non-zero terminal reward in the finite-horizon model).
   
   Define the reward-to-go $ W_t^{\sigma^{\alpha},T} $ for any agent $ i $ and strategy $ \sigma^{\alpha} $  as
   %\eq{
   \begin{equation}\label{eqr2gfh}
   W_t^{\sigma^{\alpha},T}(\mu^G_{1:t},x_{1:t}^{\alpha}) = \mE^{\sigma^{\alpha},\tsigma^{-\alpha}} \big[ \sum_{n=t}^T \delta^{n-t} R(X_n^{\alpha},A_n^{\alpha},\mu^G_n;g^{\alpha})
   %\\
   + \delta^{T+1-t} G(\mu^G_{T+1},X^{\alpha}_{T+1}) \mid \mu^G_{1:t},x_{1:t}^{\alpha} \big].
   \end{equation}
   %}
   Here agent $ i $'s strategy is $ \sigma^{\alpha} $ whereas all other agents use strategy $ \tsigma^{-\alpha} $ defined above. Since $ \mX,\mA $ are assumed to be finite and $ G $ absolutely bounded, the reward-to-go is finite $ \forall $ $ i,t,\sigma^{\alpha},\mu^G_{1:t},x_{1:t}^{\alpha} $.
   In the following, any quantity with a $T$ in the superscript refers the finite horizon model with terminal reward $G$.
   \begin{lemma}\label{thmfh1}
       For any $ t \in [T] $, $ i \in \mN $, $ \mu^G_{1:t},x_{1:t}^{\alpha} $ and $ \sigma^{\alpha} $,	
       \begin{equation} \label{eqintlem}
       V_t^{T,\alpha}(\mu^G_t,x_t^{\alpha}) \ge \mE^{\sigma^{\alpha},\tsigma^{-\alpha}} \big[ R(X_t^{\alpha},A_t^{\alpha},\mu^G_t;g^{\alpha})
       %	\\
       + \delta V_{t+1}^{T,\alpha}\big( \phi(\mu^G_t,\theta[\mu^G_t],g^{\alpha}) , X_{t+1}^{\alpha} \big) \mid \mu^G_{1:t},x_{1:t}^{\alpha} \big].
       \end{equation}
       %	}
   \end{lemma}

   The result below shows that the value function from the backwards recursive algorithm is higher than any reward-to-go.
   
   \begin{lemma}\label{thmfh2}
       For any $ t \in [T] $, $ i \in \mN $, $ \mu^G_{1:t},x_{1:t}^{\alpha} $ and $ \sigma^{\alpha} $,	
       \begin{gather}
       V_t^{T,\alpha}(\mu^G_t,x_t^{\alpha}) \ge W_t^{\sigma^{\alpha},T}(\mu^G_{1:t},x_{1:t}^{\alpha}).
       \end{gather}
   \end{lemma}
   \proof
   We use backward induction for this. At time $ T $, using the maximization property from~\eqref{eq:m_FP} (modified with terminal reward $ G $),
   \begin{subequations}
       \begin{align}
       	&V_T^{T,\alpha}(\mu^G_T,x_T^{\alpha})\\
       	&\defeq \mE^{\tilde{\gamma}_T^{\alpha,T}(\cdot \mid x_T^{\alpha}),\tilde{\gamma}_T^{-\alpha,T}} \big[ R(X_t^{\alpha},A_t^{\alpha},\mu^G_t;g^{\alpha}) + \delta G\big( \phi(\mu^G_T,\tilde{\gamma}_T^T)),X_{T+1}^{\alpha} \big) \mid \mu^G_T,x_T^{\alpha} \big]
       \\
       &\ge \mE^{{\gamma}_T^{\alpha,T}(\cdot \mid x_T^{\alpha}),\tilde{\gamma}_T^{-\alpha,T}} \big[ R(X_t^{\alpha},A_t^{\alpha},\mu^G_t;g^{\alpha})
       + \delta G\big( \phi(\mu^G_T,\tilde{\gamma}_T^T)) ,X_{T+1}^{\alpha} \big) \mid \mu^G_{1:T},x_{1:T}^{\alpha} \big]
       \\
       &= W_T^{\sigma^{\alpha},T}(h_T^{\alpha})
       \end{align}
   \end{subequations}
   Here the second inequality follows from~\eqref{eq:m_FP} and~\eqref{eq:Vdef} and the final equality is by definition in~\eqref{eqr2gfh}.
   
   Assume that the result holds for all $ n \in \{t+1,\ldots,T\} $, then at time $ t $ we have
   \begin{subequations}
       \begin{align}
       &V_t^{T,\alpha}(\mu^G_t,x_t^{\alpha})
       \\
       &\ge \mE^{\sigma_t^{\alpha},\tsigma_t^{-\alpha}} \big[ R(X_t^{\alpha},A_t^{\alpha},\mu^G_t;g^{\alpha})
       + \delta V_{t+1}^{T,\alpha}\big( \phi(\mu^G_t,\theta[\mu^G_t],g^{\alpha}) , X_{t+1}^{\alpha} \big) \mid \mu^G_{1:t},x_{1:t}^{\alpha} \big]
       \\
       &\ge \mE^{\sigma_t^{\alpha},\tsigma_t^{-\alpha}} \big[ R(X_t^{\alpha},A_t^{\alpha},\mu^G_t;g^{\alpha})
       + \delta \mE^{\sigma^{\alpha}_{t+1:T},\tsigma_{t+1:T}^{-\alpha}} \big[ \sum_{n=t+1}^T \delta^{n-(t+1)} R(X_n^{\alpha},A_n^{\alpha},\mu^G_n;g^{\alpha})
       \\ \nonumber
       &+ \delta^{T-t} G(\mu^G_{T+1},X_{T+1}^{\alpha}) \mid \mu^G_{1:t},x_{1:t}^{\alpha},\mu^G_{t+1},X_{t+1}^{\alpha} \big] \mid \mu^G_{1:t},x_{1:t}^{\alpha} \big]
       \\
       &= \mE^{\sigma^{\alpha}_{t:T},\tsigma^{-\alpha}_{t:T}} \big[ \sum_{n=t}^T \delta^{n-t} R(X_n^{\alpha},A_n^{\alpha},\mu^G_n;g^{\alpha})
       + \delta^{T+1-t}G(\mu^G_{T+1},X_{T+1}^{\alpha}) \mid \mu^G_{1:t},x_{1:t}^{\alpha} \big]
       \\
       &= W_t^{\sigma^{\alpha},T}(\mu^G_{1:t},x_{1:t}^{\alpha})
       \end{align}
   \end{subequations}
   Here the first inequality follows from Lemma~\ref{thmfh1}, the second inequality from the induction hypothesis, the third equality follows since the random variables on the right hand side do not depend on $\sigma_t^{\alpha}$, and the final equality by definition~\eqref{eqr2gfh}.
   \endproof
   
   The following result highlights the similarities between the fixed-point equation in infinite-horizon and the backwards recursion in the finite-horizon.
   
   \begin{lemma}\label{lemfhtoih}
       Consider the finite horizon game with $ G \equiv V^{\alpha} $. Then $ V_t^{T,\alpha} = V^{\alpha}$,  $ \forall $ $ i \in \mN $, $ t \in \{1,\ldots,T\} $ satisfies the backwards recursive construction stated above (adapted from \eqref{eq:m_FP} and \eqref{eq:Vdef}).

   \end{lemma}	
   \proof%[Proof of Lemma~\ref{lemfh}]
       Use backward induction for this. Consider the finite horizon algorithm at time $ t=T $, noting that $ V_{T+1}^{T,\alpha} \equiv G \equiv  V^{\alpha} $,
       \begin{subequations} \label{eqfhT}
           \begin{align} 	
           %		\eq{
           \tilde{\gamma}_T^{T,\alpha}(\cdot \mid x_T^{\alpha}) &\in \arg\max_{\gamma_T(\cdot \mid x_T^{\alpha})} \!\!\! \mE^{\gamma_T(\cdot \mid x_T^{\alpha})} \big[ R(X_t^{\alpha},A_t^{\alpha},\mu^G_t;g^{\alpha})
           %	\\
           + \delta V^{\alpha}\big( \phi(\mu^G_T,\tgamma_t^{T,\alpha}) , X_{T+1}^{\alpha} \big) \mid \mu^G_T,x_T^{\alpha} \big]
           %		}
           %		\eq{
           \\
           V_T^{T,\alpha}(\mu^G_T,x_T^{\alpha}) &= \mE^{\tilde{\gamma}_T^{T}(\cdot \mid x_T^{\alpha})} \big[ R(X_t^{\alpha},A_t^{\alpha},\mu^G_t;g^{\alpha})
           %	\\
           + \delta V\big( \phi(\mu^G_T,\tgamma_t^T) , X_{T+1}^{\alpha} \big) \mid \mu^G_T,x_T^{\alpha} \big].
           %		}
           \end{align}
       \end{subequations}
       Comparing the above set of equations with~\eqref{eq:m_FP_ih}, we can see that the pair $ (V^{\alpha},\tilde{\gamma}^{\alpha}) $ arising out of~\eqref{eq:m_FP_ih} satisfies the above. Now assume that $ V_n^{T,\alpha} \equiv V^{\alpha} $ for all $ n \in \{t+1,\ldots,T\} $. At time $ t $, in the finite horizon construction from~\eqref{eq:m_FP},~\eqref{eq:Vdef}, substituting $ V^{\alpha}$ in place of $ V_{t+1}^{T,\alpha} $ from the induction hypothesis, we get the same set of equations as~\eqref{eqfhT}. Thus $ V_t^{T,\alpha} \equiv V^{\alpha} $ satisfies it.
   \endproof
   
   \section{}
   \label{app:idih}
   \proof
   We prove this by contradiction. Suppose for the equilibrium generating function $\theta$ that generates MPE $\tsigma$, there exists $t\in[T],\alpha\in[0,1], \mu^G_{1:t}\in\cH_t^c,$ such that \eqref{eq:m_FP} is not satisfied for $\theta$
   %\footnote{Note that for $\mu^G_t \neq \mu^G_t $ for any $a_{1:t-1}$, $\phi$ can be arbitrarily defined without affecting the definition of $(\tsigma,\mu^*)$.}
   i.e. for $\tgamma_t = \theta[\mu^G_t] = \tsigma(\cdot|\mu^G_t,\cdot)$,
   \eq{
    \tilde{\gamma}^{\alpha}_t \not\in \arg\max_{\gamma^{\alpha}_t(\cdot|x_t^{\alpha})} \E^{\gamma_t(\cdot|x_t^{\alpha})} \left\{ R(X_t^{\alpha},A_t^{\alpha},\mu^G_t;g^{\alpha}) + \delta V^{\alpha}(\phi(\mu^G_t,\tilde{\gamma}_t), X_{t+1}^{\alpha}) \big\lvert  x_t^{\alpha},\mu^G_t \right\} . \label{eq:FP4}
     }
     Let $t$ be the first instance in the backward recursion when this happens. This implies $\exists\ \hat{\gamma}^{\alpha}_t$ such that
     \eq{
     \E^{\hat{\gamma}^{\alpha}_t(\cdot|x_t)} \left\{ R(X_t^{\alpha},A_t^{\alpha},\mu^G_t;g^{\alpha})+ \delta V^{\alpha}(\phi(\mu^G_t, \tilde{\gamma}^{\alpha}_t;g^{\alpha}), X_{t+1}^{\alpha}) \big\lvert  \mu^G_{1:t},x_{1:t}^{\alpha}\right\}
     \nn\\
     > \E^{\tgamma^{\alpha}_t(\cdot|x_t)} \left\{ R(X_t^{\alpha},A_t^{\alpha},\mu^G_t;g^{\alpha}) +\delta V^{\alpha}(\phi(\mu^G_t, \tilde{\gamma}^{\alpha}_t;g^{\alpha}), X_{t+1}^{\alpha}) \big\lvert  \mu^G_{1:t},x_{1:t}^{\alpha} \right\} \label{eq:E1}
     }
     This implies for $\hat{\sigma}^{\alpha}(\cdot|\mu^G_t,\cdot) = \hat{\gamma}^{\alpha}_t$,
     \eq{
     &\E^{\tsigma^{\alpha}} \left\{ \sum_{n=t}^{\infty} \delta^{n-t} R(X_n^{\alpha},A_n^{\alpha},\mu^G_n;g^{\alpha}) \big\lvert  \mu^G_{1:t-1}, x_{1:t}^{\alpha} \right\}
     \nn\\
     %&= \E^{\tsigma_t} \left\{ R_t(X_t,A_t,\mu^G_t) + \E^{\tsigma_{t:T}}  \left\{ \sum_{n=t+1}^T R_n(X_n,A_n^{\alpha},\mu^G_n) \big\lvert \mu^G_{1:t},\mu^G_{t+1}, x_{1:t},X_{t+1} \right\}  \big\vert \mu^G_{1:t}, x_{1:t} \right\}% \label{eq:E2a}
   \\
     &= \E^{\tsigma_t^{\alpha},\tsigma_t^{-\alpha}} \left\{ R(X_t^{\alpha},A_t^{\alpha},\mu^G_t;g^{\alpha}) + \right.\nn\\
   &\left.\E^{\tsigma_{t+1:T}^{\alpha} \tsigma_{t+1:T}^{-\alpha}}   \left\{ \sum_{n=t+1}^{\infty} \delta^{n-t}R(X_n^{\alpha},A_n^{\alpha},\mu^G_n;g^{\alpha}) \big\lvert \mu^G_{1:t-1},\mu^G_{t+1}, x_{1:t}^{\alpha},X_{t+1}^{\alpha} \right\}  \big\vert \mu^G_{1:t}, x_{1:t}^{\alpha} \right\} \label{eq:E2}
     \\
     &=\E^{\tgamma^{\alpha}_t(\cdot|x_t) \tilde{\gamma}^{-\alpha}_t} \left\{ R(X_t^{\alpha},A_t^{\alpha},\mu^G_t;g^{\alpha}) +\delta V^{\alpha} (\phi(\mu^G_t, \tilde{\gamma}^{\alpha}_t;g^{\alpha}), X_{t+1}^{\alpha}) \big\lvert  \mu^G_t,x_t^{\alpha} \right\} \label{eq:E3}
     \\
     &< \E^{\hat{\sigma}^{\alpha}_t(\cdot|\mu^G_t,x_t^{\alpha}) \tilde{\gamma}^{-\alpha}_t} \left\{ R(X_t^{\alpha},A_t^{\alpha},\mu^G_t;g^{\alpha}) +\delta V^{\alpha}(\phi(\mu^G_t, \tilde{\gamma}^{\alpha}_t;g^{\alpha}), X_{t+1}^{\alpha}) \big\lvert  \mu^G_t,x_t^{\alpha} \right\}\label{eq:E4}
     \\
     &= \E^{\hat{\sigma}^{\alpha}_t \tsigma_t^{-\alpha}} \left\{ R(X_t^{\alpha},A_t^{\alpha},\mu^G_t;g^{\alpha}) +  \right.\nn\\
   &\left.\E^{\tsigma_{t+1:T}^{\alpha} \tsigma_{t+1:T}^{-\alpha}}\left\{ \sum_{n=t+1}^{\infty} \delta^{n-t}R(X_n^{\alpha},A_n^{\alpha},\mu^G_n;g^{\alpha}) \big\lvert \mu^G_{1:t},\mu^G_{t+1}, x_{1:t}^{\alpha},X_{t+1}^{\alpha}\right\} \big\vert \mu^G_{1:t}, x_{1:t}^{\alpha} \right\}\label{eq:E5}
     \\
     &=\E^{\hat{\sigma}^{\alpha}_t,\tsigma_{t+1:T}^{\alpha} \tsigma_{t:T}^{-\alpha}} \left\{ \sum_{n=t}^{\infty} \delta^{n-t} R(X_n^{\alpha},A_n^{\alpha},\mu^G_n;g^{\alpha}) \big\lvert  \mu^G_{1:t}, x_{1:t}^{\alpha} \right\},\label{eq:E6}
     }
     where \eqref{eq:E3} follows from the definitions of $\tgamma_t$ and Appendix~\ref{app:D}, \eqref{eq:E4} follows from \eqref{eq:E1} and the definition of $\hat{\sigma}_t$, \eqref{eq:E5} follows from Appendix~\ref{app:D}. However, this leads to a contradiction since $\tsigma$ is a \gmfe of the game.
   \endproof
   
   % \end{APPENDICES}

\medskip
\small
\bibliographystyle{IEEEtran}

\bibliography{library,deepanshu,Reference1}

\begin{thebibliography}{10}
\providecommand{\url}[1]{#1}
\csname url@samestyle\endcsname
\providecommand{\newblock}{\relax}
\providecommand{\bibinfo}[2]{#2}
\providecommand{\BIBentrySTDinterwordspacing}{\spaceskip=0pt\relax}
\providecommand{\BIBentryALTinterwordstretchfactor}{4}
\providecommand{\BIBentryALTinterwordspacing}{\spaceskip=\fontdimen2\font plus
\BIBentryALTinterwordstretchfactor\fontdimen3\font minus
  \fontdimen4\font\relax}
\providecommand{\BIBforeignlanguage}[2]{{%
\expandafter\ifx\csname l@#1\endcsname\relax
\typeout{** WARNING: IEEEtran.bst: No hyphenation pattern has been}%
\typeout{** loaded for the language `#1'. Using the pattern for}%
\typeout{** the default language instead.}%
\else
\language=\csname l@#1\endcsname
\fi
#2}}
\providecommand{\BIBdecl}{\relax}
\BIBdecl

\bibitem{MaTi01}
E.~Maskin and J.~Tirole, ``Markov perfect equilibrium: I. observable actions,''
  \emph{Journal of Economic Theory}, vol. 100, no.~2, pp. 191--219, 2001.

\bibitem{ErPa95}
R.~Ericson and A.~Pakes, ``Markov-perfect industry dynamics: A framework for
  empirical work,'' \emph{The Review of Economic Studies}, vol.~62, no.~1, pp.
  53--82, 1995.

\bibitem{BeVa96}
D.~Bergemann and J.~V{\"a}lim{\"a}ki, ``Learning and strategic pricing,''
  \emph{Econometrica: Journal of the Econometric Society}, pp. 1125--1149,
  1996.

\bibitem{AcRo01}
D.~Acemo\u{g}lu and J.~A. Robinson, ``A theory of political transitions,''
  \emph{American Economic Review}, pp. 938--963, 2001.

\bibitem{HuMaCa06}
M.~Huang, R.~P. Malham{\'e}, and P.~E. Caines, ``Large population stochastic
  dynamic games: closed-loop mckean-vlasov systems and the nash certainty
  equivalence principle,'' \emph{Communications in Information \& Systems},
  vol.~6, no.~3, pp. 221--252, 2006.

\bibitem{LaLi07}
J.-M. Lasry and P.-L. Lions, ``Mean field games,'' \emph{Japanese Journal of
  Mathematics}, vol.~2, no.~1, pp. 229--260, 2007.

\bibitem{Caetal15}
P.~Cardaliaguet, F.~Delarue, J.-M. Lasry, and P.-L. Lions, ``The master
  equation and the convergence problem in mean field games,'' \emph{arXiv
  preprint arXiv:1509.02505}, 2015.

\bibitem{La16}
D.~Lacker, ``A general characterization of the mean field limit for stochastic
  differential games,'' \emph{Probability Theory and Related Fields}, vol. 165,
  no. 3-4, pp. 581--648, 2016.

\bibitem{Fi17}
M.~Fischer \emph{et~al.}, ``On the connection between symmetric $ n $-player
  games and mean field games,'' \emph{The Annals of Applied Probability},
  vol.~27, no.~2, pp. 757--810, 2017.

\bibitem{La18}
D.~Lacker, ``On the convergence of closed-loop nash equilibria to the mean
  field game limit,'' \emph{arXiv preprint arXiv:1808.02745}, 2018.

\bibitem{DeLaRa19}
\BIBentryALTinterwordspacing
F.~Delarue, D.~Lacker, and K.~Ramanan, ``From the master equation to mean field
  game limit theory: a central limit theorem,'' \emph{Electron. J. Probab.},
  vol.~24, p. 54 pp., 2019. [Online]. Available:
  \url{https://doi.org/10.1214/19-EJP298}
\BIBentrySTDinterwordspacing

\bibitem{PaOz19}
F.~Parise and A.~Ozdaglar, ``Graphon games,'' in \emph{Proceedings of the 2019
  ACM Conference on Economics and Computation}, 2019, pp. 457--458.

\bibitem{Lo12}
L.~Lov{\'a}sz, \emph{Large networks and graph limits}.\hskip 1em plus 0.5em
  minus 0.4em\relax American Mathematical Soc., 2012, vol.~60.

\bibitem{CaHu18}
P.~E. Caines and M.~Huang, ``Graphon mean field games and the gmfg equations,''
  in \emph{2018 IEEE Conference on Decision and Control (CDC)}.\hskip 1em plus
  0.5em minus 0.4em\relax IEEE, 2018, pp. 4129--4134.

\bibitem{NaMaTe13}
A.~Nayyar, A.~Mahajan, and D.~Teneketzis, ``Decentralized stochastic control
  with partial history sharing: A common information approach,''
  \emph{Automatic Control, IEEE Transactions on}, vol.~58, no.~7, pp.
  1644--1658, 2013.

\end{thebibliography}


\begin{thebibliography}{10}
\providecommand{\url}[1]{#1}
\csname url@samestyle\endcsname
\providecommand{\newblock}{\relax}
\providecommand{\bibinfo}[2]{#2}
\providecommand{\BIBentrySTDinterwordspacing}{\spaceskip=0pt\relax}
\providecommand{\BIBentryALTinterwordstretchfactor}{4}
\providecommand{\BIBentryALTinterwordspacing}{\spaceskip=\fontdimen2\font plus
\BIBentryALTinterwordstretchfactor\fontdimen3\font minus
  \fontdimen4\font\relax}
\providecommand{\BIBforeignlanguage}[2]{{%
\expandafter\ifx\csname l@#1\endcsname\relax
\typeout{** WARNING: IEEEtran.bst: No hyphenation pattern has been}%
\typeout{** loaded for the language `#1'. Using the pattern for}%
\typeout{** the default language instead.}%
\else
\language=\csname l@#1\endcsname
\fi
#2}}
\providecommand{\BIBdecl}{\relax}
\BIBdecl

\bibitem{VaMiVi21}
\BIBentryALTinterwordspacing
D.~Vasal, R.~K. Mishra, and S.~Vishwanath, ``{Sequential decomposition of
  graphon mean field games},'' \emph{Proceedings of the American Control
  Conference}, vol. 2021-May, pp. 730--736, jan 2020. [Online]. Available:
  \url{https://arxiv.org/abs/2001.05633v1}
\BIBentrySTDinterwordspacing

\bibitem{Wi68}
H.~Witsenhausen, ``A counterexample in stochastic optimum control,'' \emph{SIAM
  Journal on Control}, vol.~6, no.~1, pp. 131--147, 1968.

\bibitem{NaMaTe13}
A.~Nayyar, A.~Mahajan, and D.~Teneketzis, ``Decentralized stochastic control
  with partial history sharing: A common information approach,''
  \emph{Automatic Control, IEEE Transactions on}, vol.~58, no.~7, pp.
  1644--1658, 2013.

\bibitem{AbMa14}
\BIBentryALTinterwordspacing
J.~Arabneydi and A.~Mahajan, ``{Team Optimal Control of Coupled Subsystems with
  Mean-Field Sharing},'' dec 2020. [Online]. Available:
  \url{https://arxiv.org/abs/2012.01418v1}
\BIBentrySTDinterwordspacing

\bibitem{MaTi01}
E.~Maskin and J.~Tirole, ``Markov perfect equilibrium: I. observable actions,''
  \emph{Journal of Economic Theory}, vol. 100, no.~2, pp. 191--219, 2001.

\bibitem{ErPa95}
R.~Ericson and A.~Pakes, ``Markov-perfect industry dynamics: A framework for
  empirical work,'' \emph{The Review of Economic Studies}, vol.~62, no.~1, pp.
  53--82, 1995.

\bibitem{BeVa96}
D.~Bergemann and J.~V{\"a}lim{\"a}ki, ``Learning and strategic pricing,''
  \emph{Econometrica: Journal of the Econometric Society}, pp. 1125--1149,
  1996.

\bibitem{AcRo01}
D.~Acem\u{o}glu and J.~A. Robinson, ``{A theory of political
  transitions},'' \emph{American Economic Review}, pp. 938--963, 2001.

\bibitem{VaSiAn16arxiv}
D.~Vasal, A.~Sinha, and A.~Anastasopoulos, ``A systematic process for
  evaluating structured perfect bayesian equilibria in dynamic games with
  asymmetric information,'' \emph{IEEE Transactions on Automatic Control},
  2018.

\bibitem{VaAn16}
D.~Vasal and A.~Anastasopoulos, ``A systematic process for evaluating
  structured perfect {B}ayesian equilibria in dynamic games with asymmetric
  information,'' in \emph{American {C}ontrol {C}onference}, Boston, US, 2016,
  available on arXiv.

\bibitem{Ta17}
H.~T. Jahormi, ``On design and analysis of cyber-physical systems with
  strategic agents,'' Ph.D. dissertation, University of Michigan, Ann Arbor,
  2017.

\bibitem{HeAn20}
\BIBentryALTinterwordspacing
N.~Heydaribeni and A.~Anastasopoulos, ``{Structured Equilibria for Dynamic
  Games with Asymmetric Information and Dependent Types},'' sep 2020. [Online].
  Available: \url{https://arxiv.org/abs/2009.04253v1}
\BIBentrySTDinterwordspacing

\bibitem{OuTaTe17}
Y.~Ouyang, H.~Tavafoghi, and D.~Teneketzis, ``Dynamic games with asymmetric
  information: Common information based perfect bayesian equilibria and
  sequential decomposition,'' \emph{IEEE Transactions on Automatic Control},
  vol.~62, no.~1, pp. 222--237, 2017.

\bibitem{HuMaCa06}
M.~Huang, R.~P. Malham{\'e}, and P.~E. Caines, ``Large population stochastic
  dynamic games: closed-loop mckean-vlasov systems and the nash certainty
  equivalence principle,'' \emph{Communications in Information \& Systems},
  vol.~6, no.~3, pp. 221--252, 2006.

\bibitem{LaLi07}
J.-M. Lasry and P.-L. Lions, ``Mean field games,'' \emph{Japanese Journal of
  Mathematics}, vol.~2, no.~1, pp. 229--260, 2007.

\bibitem{Caetal15}
P.~Cardaliaguet, F.~Delarue, J.-M. Lasry, and P.-L. Lions, ``The master
  equation and the convergence problem in mean field games,'' \emph{arXiv
  preprint arXiv:1509.02505}, 2015.

\bibitem{La16}
D.~Lacker, ``A general characterization of the mean field limit for stochastic
  differential games,'' \emph{Probability Theory and Related Fields}, vol. 165,
  no. 3-4, pp. 581--648, 2016.

\bibitem{Fi17}
M.~Fischer \emph{et~al.}, ``On the connection between symmetric $ n $-player
  games and mean field games,'' \emph{The Annals of Applied Probability},
  vol.~27, no.~2, pp. 757--810, 2017.

\bibitem{La18}
D.~Lacker, ``On the convergence of closed-loop nash equilibria to the mean
  field game limit,'' \emph{arXiv preprint arXiv:1808.02745}, 2018.

\bibitem{DeLaRa19}
\BIBentryALTinterwordspacing
F.~Delarue, D.~Lacker, and K.~Ramanan, ``From the master equation to mean field
  game limit theory: a central limit theorem,'' \emph{Electron. J. Probab.},
  vol.~24, p. 54 pp., 2019. [Online]. Available:
  \url{https://doi.org/10.1214/19-EJP298}
\BIBentrySTDinterwordspacing

\bibitem{PaOz19}
F.~Parise and A.~Ozdaglar, ``Graphon games,'' in \emph{Proceedings of the 2019
  ACM Conference on Economics and Computation}, 2019, pp. 457--458.

\bibitem{Lo12}
L.~Lov{\'a}sz, \emph{Large networks and graph limits}.\hskip 1em plus 0.5em
  minus 0.4em\relax American Mathematical Soc., 2012, vol.~60.

\bibitem{CaHu18}
P.~E. Caines and M.~Huang, ``Graphon mean field games and the gmfg equations,''
  in \emph{2018 IEEE Conference on Decision and Control (CDC)}.\hskip 1em plus
  0.5em minus 0.4em\relax IEEE, 2018, pp. 4129--4134.

\bibitem{CaDeLaLi15}
\BIBentryALTinterwordspacing
P.~Cardaliaguet, F.~Delarue, J.-M. Lasry, and P.-L. Lions, ``{The master
  equation and the convergence problem in mean field games},'' \emph{Annals of
  Mathematics Studies}, vol. 2019-Janua, no. 201, pp. 1--222, sep 2015.
  [Online]. Available: \url{https://arxiv.org/abs/1509.02505v1}
\BIBentrySTDinterwordspacing

\bibitem{VaMiVi20}
R.~Mishra, D.~Vasal, and S.~Vishwanath, ``{Model-free Reinforcement Learning
  for Stochastic Stackelberg Security Games},'' 2020.

\bibitem{TcCaHu20}
R.~F. Tchuendom, P.~E. Caines, and M.~Huang, ``{On the Master Equation for
  Linear Quadratic Graphon Mean Field Games},'' \emph{Proceedings of the IEEE
  Conference on Decision and Control}, vol. 2020-Decem, pp. 1026--1031, dec
  2020.

\bibitem{KuVa86}
P.~Kumar and P.~Varaiya, ``Stochastic systems,'' 1986.

\end{thebibliography}
\end{document}